\documentclass[notitlepage,a4,10.5pt]{article}

\usepackage{graphicx}

\usepackage{amsmath}%
\usepackage{amsfonts}%
\usepackage{amssymb}%
\usepackage{mathrsfs}
\usepackage{amsthm}
\usepackage{float}

\usepackage{authblk}

\usepackage[left=2.5cm,right=2.5cm,top=3cm,bottom=3cm]{geometry} 

\usepackage{xcolor}

\newcommand{\mc}{\mathcal}
\newcommand{\cp}{\times}

\newcommand{\Norm}[1]{\left\lvert\left\lvert #1 \right\rvert\right\rvert}

\newcommand{\bol}{\boldsymbol}

\newcommand{\abs}[1]{\left\lvert{#1}\right\rvert}

\newcommand{\lr}[1]{\left({#1}\right)}
\newcommand{\lrs}[1]{\left[{#1}\right]}
\newcommand{\lrc}[1]{\left\{{#1}\right\}}

\newcommand{\mf}{\mathfrak}
\newcommand{\p}{\partial}

\newcommand{\ti}[1]{\textit{#1}}
\newcommand{\tb}[1]{\textbf{#1}}

\newtheorem{remark}{\textit{Remark}}
\newtheorem{theorem}{\textit{Theorem}}
\newtheorem{proposition}{\textit{Proposition}}

\begin{document}

\title{
A Reduced Ideal MHD System 
for Nonlinear
Magnetic Field Turbulence in Plasmas with 
Approximate Flux Surfaces 
}
\author[1]{Naoki Sato} 
\author[2]{Michio Yamada}
\affil[1]{National Institute for Fusion Science, \protect\\ 322-6 Oroshi-cho Toki-city, Gifu 509-5292, Japan \protect\\ Email: sato.naoki@nifs.ac.jp}
\affil[2]{Research Institute for Mathematical Sciences, \protect\\ Kyoto University, Kyoto 606-8502, Japan
\protect \\ Email: yamada@kurims.kyoto-u.ac.jp}
\date{\today}
\setcounter{Maxaffil}{0}
\renewcommand\Affilfont{\itshape\small}

    \maketitle
    \begin{abstract}
    This paper studies the nonlinear evolution of magnetic field turbulence in proximity of steady ideal MHD configurations  characterized by a small electric current, a small plasma flow, and approximate flux surfaces, a physical setting that is relevant for  plasma confinement in stellarators. The aim is to
    gather insight on magnetic field dynamics, to elucidate accessibility and stability of three-dimensional MHD equilibria, as well as to formulate practical methods to compute them. 
    Starting from the ideal MHD equations, a reduced dynamical system of two coupled nonlinear PDEs for the flux function and the angle variable associated with the Clebsch representation of the magnetic field is obtained. It is shown that under suitable boundary and gauge conditions such reduced system preserves magnetic energy, magnetic helicity, and total magnetic flux. 
    The noncanonical Hamiltonian structure of the reduced system is identified, 
    and used to show the nonlinear stability of steady solutions against perturbations involving only one Clebsch potential. 
    The Hamiltonian structure is also applied to 
    construct a dissipative dynamical system through the method of double brackets. 
    This dissipative system enables the computation of MHD equilibria by minimizing energy until a critical point of the Hamiltonian is reached. 
    Finally, an iterative scheme based on the alternate solution of the two steady equations in the reduced system is proposed as a further method to compute MHD equilibria. 
    A theorem is proven which states that the iterative scheme converges to a nontrivial MHD equilbrium as long as solutions exist at each step of the iteration. 
    \end{abstract}






\section{Introduction}
This study is concerned with the dynamics of the magnetic field around a magnetohydrodynamics (MHD) equilibrium \cite{Kruskal}
\begin{equation}
\lr{\nabla\cp\bol{B}}\cp\bol{B}=\mu_0\nabla P,~~~~\nabla\cdot\bol{B}=0~~~~{\rm in}~~\Omega.\label{MHDE}  
\end{equation}
In this equation $\bol{B}\lr{\bol{x}}$ denotes the equilibrium magnetic field in a smooth bounded domain $\Omega\subset\mathbb{R}^3$, $\bol{x}=\lr{x,y,z}$ Cartesian coordinates, $\mu_0$ the vacuum permeability, and $P\lr{\bol{x}}$ the equilibrium pressure field. For the purpose of the present paper, 
the dynamics around \eqref{MHDE} is governed by the ideal MHD equations in $\Omega$, 
\begin{subequations}
\begin{align}
\rho\frac{\p\bol{u}}{\p t}=&-\rho\bol{u}\cdot\nabla\bol{u}+\frac{1}{\mu_0}\lr{\nabla\cp\bol{B}}\cp\bol{B}-{\nabla P},\label{MHD1-1}\\
\frac{\p\bol{B}}{\p t}=&-\nabla\cp\bol{E},\label{MHD1-2}\\
\frac{\p\rho}{\p t}=&-\nabla\cdot\lr{\rho\bol{u}},\\
\nabla\cdot\bol{B}=&0.
\end{align}\label{MHD1}
\end{subequations}
Here, $\bol{u}\lr{\bol{x},t}$, 
$\bol{B}\lr{\bol{x},t}$, $\rho\lr{\bol{x},t}$,  $P\lr{\bol{x},t}$, and $\bol{E}\lr{\bol{x},t}$ denote the time-dependent velocity field, magnetic field,  plasma mass density, pressure field, and electric field respectively.
The closure of system \eqref{MHD1} can be obtained by choosing an equation of state relating $P$ to the other fields, and by determining $\bol{E}$ from the electron fluid momentum equation. These aspects will be discussed in detail later.  
The MHD equilibrium \eqref{MHDE} can be obtained from system \eqref{MHD1} by setting time derivatives to zero and taking a vanishing plasma flow $\bol{u}=\bol{0}$.  

Despite its relevance for the confinement of magnetized plasmas  
and the development of nuclear fusion reactors known as stellarators, 
a general theory concerning the existence of regular solutions $\lr{\bol{B},P}$ of the MHD equilibrium equations \eqref{MHDE} is not available at present \cite{LoSurdo}. 
This is because the characteristic surfaces associated with the first-order system of PDEs \eqref{MHDE} depend on the unknown $\bol{B}$ \cite{YosYam,GradRed}. 
The existence of regular solutions is known for the special cases in which the pressure is constant or the fields $\lr{\bol{B},P}$ are invariant under some combination of continuous Euclidean isometries of $\mathbb{R}^3$ (rotations and translations). In the first case the magnetic field is a Beltrami field \cite{YosGiga}. In the second case, equation \eqref{MHDE} reduces to the Grad-Shafranov equation, a nonlinear second order elliptic PDE for the flux function  \cite{GradToroidal,Eden1,Eden2}. 
Both cases are however not relevant for stellarators \cite{Hel}, 
which consist of toroidal vessels without trivial symmetries 
surrounded by coils with complex shapes whose purpose is to 
generate the field line twist required to 
minimize particle losses caused by cross-field drifts. 
In principle, stellarators achieve steady plasma confinement mostly through an externally produced vacuum magnetic field, and are therefore more suitable for continued operation compared to an axially symmetric tokamak where the field line twist is obtained by driving an electric current within the plasma.  
However, the lack of axial symmetry results in the breaking of  
conservation of vertical angular momentum, a fact that deteriorates the quality of plasma confinement. For this reason, in addition to \eqref{MHDE} the equilibrium magnetic field within a stellarator must satisfy additional conditions, such as quasisymmetry, a property that ensures   
particle confinement by constraining particle orbits close to a given flux surface \cite{Rod21,Landre,SatoSciRep}. 

In this context, the aim of the present paper is (i) to 
obtain a closed set of reduced equations preserving the Hamiltonian structure \cite{MorrisonRMP,Abdel} of ideal MHD and describing the nonlinear evolution of the magnetic field 
in proximity of MHD equilibria \eqref{MHDE} and in a physical regime relevant for stellarator plasmas, (ii)   
to use the derived equations to elucidate the stability properties of such equilibria,  
and (iii) to apply the derived equations to formulate dissipative and iterative schemes to construct nontrivial 
MHD equilibria \eqref{MHDE} with nested flux surfaces and a non-vanishing pressure gradient in toroidal domains of arbitrary shape. 
In addition to providing a toy model of magnetic field turbulence in a physically relevant setting, we conjecture that the derived results may 
serve as a starting point for a mathematical proof of 
existence of nontrivial MHD equilibria \eqref{MHDE} in toroidal domains without Euclidean symmetries. 


The present paper is organized as follows.
In section 2 the ideal MHD equations \eqref{MHD1} 
are reduced according to 
an ordering in which 
the plasma flow $\bol{u}$, the electric current $\nabla\cp\bol{B}$, and time derivatives $\p/\p t$ are small, and the pressure field $P$ is related to the mass density $\rho$ and the velocity field $\bol{u}$ by a generalized Bernoulli principle.  
This ordering also implies the existence of approximate flux surfaces $\Psi$ for the magnetic field. This fact is used to enforce at leading order a Clebsch representation  \cite{YosClebsch,YosEpi2D} of the magnetic field  $\bol{B}=\nabla\Psi\cp\nabla\Theta$ with $\Psi$ the flux function and $\Theta$ a   multi-valued (angle) variable. 
In the reduced system, the dynamics of the magnetic field is thus described by a pair of coupled equations for the 
two Clebsch potentials $\Psi$ and $\Theta$.

In section 3 it is shown that under suitable boundary conditions and gauge conditions for the magnetic vector potential the reduced dynamics preserves magnetic energy, magnetic helicity, and total magnetic flux. 
These conservation laws are then applied to describe steady states in terms of critical points of a   functional of $\Psi$ and $\Theta$ given by a linear combination of magnetic energy and total magnetic flux. This variational formulation can be physically interpreted  in analogy with  
Taylor relaxation \cite{Taylor,Taylor2}  
in which magnetic energy is minimized under the constraint of magnetic helicity \cite{Wolt}. Here,  functionals involving higher order derivatives of the dynamical variables are dissipated at a faster rate by non-ideal (dissipative) mechanisms \cite{YosMahPRL}.

In section 4 the noncanonical Hamiltonian  structure of the reduced equations is identified 
in terms of a Poisson bracket and a Hamiltonian functional \cite{Morrison89,Little89}. 
In particular, it is shown that the Poisson bracket  satisfies all the Poisson bracket axioms, including the Jacobi identity \cite{Olver}.

In section 5 the Hamiltonian structure obtained in section 4 is used to prove 
that steady solutions of the reduced dynamics are nonlinearly stable \cite{Holm,Tronci,Rein,Arnold} against turbulent fluctuations involving only one of the two Clebsch potentials. Here, we recall that
a positive second variation of the Hamiltonian is not sufficient to guarantee nonlinear stability. In general, a norm on the space of solutions must be found such that the deviation of the perturbed solution at a given time 
is bound by the discrepancy 
of initial conditions in the prescribed topology. 
The type of nonlinear stability shown here effectively constrains the 
deviation   
of the perturbed Clebsch potential from initial conditions  
on the level sets of the other unperturbed Clebsch potential.

In section 6 the method of double brackets \cite{PJMdb1,PJMdb2} 
is used to formulate a dissipative dynamical systems for the Clebsch potentials $\Psi$ and $\Theta$ with the property that, instead of being constant, the Hamiltonian is progressively dissipated. 
This pair of equations has the structure of coupled diffusion equations. 
Double bracket dynamics is  obtained by applying twice  the Poisson bracket  
and represents an effective tool to compute energy minima while preserving the 
Casimir invariants that span the kernel of the Poisson bracket \cite{Furukawa,Vallis}. 
The derived dissipative dynamical system may therefore be applied to compute MHD equilibria \eqref{MHDE} corresponding to critical points of the Hamiltonian in a dynamical fashion. 

In section 7 a  second iterative scheme based on the alternate solution of the two steady equations of the reduced dynamical system obtained in section 2 is discussed. Here, a theorem is proven 
which states that the iterative scheme converges to a nontrivial MHD
equilbrium (with a non-vanishing pressure gradient)  as long as solutions exist at each step of the iteration.
The dissipative and iterative schemes obtained in sections 6 and 7 may be regarded as alternative methods to the  
constrained minimizaition of the plasma energy $\int_{\Omega}\lr{\frac{\bol{B}^2}{2\mu_0}+\frac{P}{\gamma-1}}dV$, where $\gamma\geq 0$ is the adiabatic index,  
often used to numerically  compute MHD equilibria \cite{Chodura,Hir}. 

Concluding remarks are given in section 8.








\section{A reduced ideal MHD system for magnetic field turbulence in plasmas with small currents and approximate flux surfaces}
The aim of this section is to derive a reduction of the ideal MHD equations 
that determines the nonlinear evolution of 
magnetic field turbulence within a plasma 
characterized by approximate flux surfaces. 
This system should be appropriate to describe the confinement regime of a tokamak or a stellarator 
provided that flux surfaces exist to some degree throughout time evolution. 
This statement will be made quantitatively precise later. 


We start by considering the ideal MHD equations \eqref{MHD1} 
in a smooth bounded domain $\Omega\subset\mathbb{R}^3$. 
The relationship between the pressure field $P$ and the other field variables will be given later in the form of a generalized Bernoulli principle. 
The relation between the electric field $\bol{E}$ and the other fields is given by a generalized Ohm's law following from the electron fluid momentum equation
\begin{equation}
m_en_e\frac{d\bol{u}_e}{dt}=-e n_e\lr{\bol{E}+\bol{u}_e\cp\bol{B}}-\nabla P_e.\label{emeq}
\end{equation}
Here, $m_e$, $-e$, $n_e\lr{\bol{x},t}$, $\bol{u}_e\lr{\bol{x},t}$, and $P_e\lr{\bol{x},t}$ denote the electron
mass, charge, number density, fluid velocity, and pressure field, and we defined $d/dt=\p/\p t+\bol{u}_e\cdot\nabla$. 
Now recall that 
\begin{equation}
\bol{u}=\frac{\bol{u}_i+\delta\bol{u}_e}{1+\delta},~~~~\rho=m_in\lr{1+\delta},~~~~\bol{u}_e=\bol{u}-\frac{\nabla\cp\bol{B}}{e\mu_0\lr{1+\delta} n},\label{eq3}
\end{equation}
where $m_i$, $\delta=m_e/m_i$, $n_i$, $\bol{u}_i\lr{\bol{x},t}$ are the ion mass, the electron to ion mass ratio, the ion number density, and the ion fluid velocity, and we have assumed quasineutrality $n_i=n_e=n$.
The electron momentum equation \eqref{emeq} therefore gives the generalized Ohm's law
\begin{equation}
\bol{E}=\bol{B}\cp\bol{u}+\frac{m_i}{e\mu_0\rho}\lr{\nabla\cp\bol{B}}\cp\bol{B}+\frac{m_i\lr{1+\delta}}{e\rho}\nabla P_e-\frac{m_e}{e}\frac{d\bol{u}_e}{dt}. \label{E}
\end{equation}
The closure of the MHD system \eqref{MHD1} can thus be obtained by neglecting the last term involving the electron inertia, and by a proper ansatz on the electron pressure $P_e$. 
Indeed, assuming a barotropic equation of state for the electron pressure,
$P_e=P_e\lr{\rho}$, we obtain 
\begin{equation}
\nabla\cp\bol{E}=\nabla\cp\lrs{\lr{\frac{\kappa}{\rho}\nabla\cp\bol{B}-\bol{u}}\cp\bol{B}}, \label{GOL}
\end{equation}
with $\kappa=m_i/e\mu_0$ a physical constant associated with the Hall effect. 
We also demand that the vector fields $\bol{u}$,  $\bol{B}$, and $\nabla\cp\bol{B}$, and the pressure $P$ satisfy the boundary conditions
\begin{equation}
\bol{u}\cdot\bol{n}=0,~~~~\bol{B}\cdot\bol{n}=0,~~~~\nabla\cp\bol{B}\cdot\bol{n}=0,~~~~P={\rm constant}~~~~{\rm on}~~\p\Omega, \label{bc}
\end{equation}
where $\p\Omega$ denotes the boundary of $\Omega$ and $\bol{n}$ the unit outward normal to $\p\Omega$. Notice in particular that the boundary condition for $\nabla\cp\bol{B}$ implies that there is no net current flow across the bounding surface $\p\Omega$, and it is expected to hold true as long as both $\bol{u}_e$ 
and $\bol{u}_i$ are tangent to $\p\Omega$ (recall equation \eqref{eq3}).
We also emphasize that the boundary conditions \eqref{bc} describe what we expect from a physical standpoint. The set of boundary conditions required for the existence of solutions will be described for each set of governing equations when necessary. 

Let $\epsilon>0$ denote a small ordering parameter, $L\sim\Omega^{1/3}$ the characteristic size of the system (e.g. the linear size of a stellarator), and $T$ a reference time scale (for example, a small fraction of the confinement time scale). Assuming $\rho>0$, we start by considering the following ordering conditions 
\begin{subequations}
\begin{align}
&{\frac{T}{L\sqrt{\mu_0\rho}}}\bol{B}\sim 1,\\
&\frac{T}{L}\bol{u}\sim T\nabla\cp\bol{u}\sim \epsilon,\\
&{\frac{T}{\sqrt{\mu_0\rho}}}\nabla\cp\bol{B}\sim \frac{T^2}{\rho 
 L^2}P\sim \frac{T}{\abs{\bol{u}}}\frac{\p\bol{u}}{\p t}\sim \frac{T}{\rho}\frac{\p\rho}{\p t}\sim \frac{T}{\abs{\bol{B}}}\frac{\p\bol{B}}{\p t}\sim \frac{T}{\abs{\bol{B}}}\nabla\cp\bol{E}\sim\epsilon^2
\end{align}\label{ordering}
\end{subequations}
Physically, these requirements 
describe a magnetized plasma regime with small flow and small electric current and where the velocity field and the plasma density evolve slowly  
in response to the turbulent evolution of the magnetic field, which is driven by the non-vanishing of $\nabla\cp\bol{E}$ in equation \eqref{MHD1-2}. 
This regime is relevant for example within stellarators, which 
are designed to minimize internal flows and currents. 
Alternative orderings leading to the same governing equations 
will be described at the end of this section. 
Later we will also see that the smallness of $\nabla\cp\bol{E}$ implies that the magnetic field is endowed with approximate flux surfaces.  
Note that $\abs{\bol{B}}/\sqrt{\mu_0\rho}$ is the Alfv\'en speed, and that
the ordering \eqref{ordering} does not involve conditions on the size of the typical particle number $\rho L^3/m_i$. 
We now make the following ordering prescription, 
\begin{equation}
\frac{T^3}{L^2}\bol{u}\cdot\lr{\frac{1}{\rho}\nabla P+\frac{1}{2}
\nabla\bol{u}^2}\sim \epsilon^4.\label{EoS}
\end{equation}
When the density $\rho=\rho_c\in\mathbb{R}$ is constant 
this relationship can be satisfied through the  Bernoulli principle $\nabla\lr{{P}+\rho_c\bol{u}^2/2+h}=\bol{0}$, where $h$ is any function such that $\bol{u}\cdot\nabla h=0$.
Hence, equation \eqref{EoS}, 
which plays a role analogous to an equation of state, can be interpreted as a generalized Bernoulli principle. 
Notice also that equation \eqref{EoS} arises naturally from the plasma momentum equation \eqref{MHD1-1} when the system is steady and the vorticity $\nabla\cp\bol{u}$ and the electric current $\nabla\cp\bol{B}$ are small (a regime that is particularly relevant for stellarators).
Equations \eqref{GOL} and \eqref{EoS} and the ordering conditions \eqref{ordering} can then be used to reduce system \eqref{MHD1} to the leading order system 
\begin{subequations}
\begin{align}
\rho\lr{\nabla\cp\bol{u}}\cp\bol{u}=&\frac{1}{\mu_0}\lr{\nabla\cp\bol{B}}\cp\bol{B}-\nabla P-\frac{1}{2}\rho\nabla\bol{u}^2,\label{MHD2-1}\\
\bol{u}\cdot\lr{\nabla P+\frac{1}{2}\rho\nabla\bol{u}^2}=&0,\label{MHD2-2}\\
\frac{\p\bol{B}}{\p t}=&\nabla\cp\lrs{\lr{\bol{u}-\frac{\kappa}{\rho}\nabla\cp\bol{B}}\cp\bol{B}},\label{MHD2-3}\\
\nabla\cdot\lr{\rho\bol{u}}=&0,\label{MHD2-4}\\
\nabla\cdot\bol{B}=&0.\label{MHD2-5}
\end{align}\label{MHD2}
\end{subequations}
Observe that in the limit $\bol{u}\rightarrow\bol{0}$ steady solutions of system \eqref{MHD2} are described by the MHD equilibrium equations \eqref{MHDE} and a  barotropic equation of state $P=P\lr{\rho}$.
We also emphasize that, in addition to the induction equation \eqref{MHD2-3}  
and the divergence-free condition \eqref{MHD2-5} for the magnetic field, 
system \eqref{MHD2} comprises 
5 other equations (momentum equation \eqref{MHD2-1}, equation of state \eqref{MHD2-2}, and continuity equation \eqref{MHD2-4}) 
    involving 5 additional variables $\lr{\bol{u},\rho,P}$. 
In particular, given the magnetic field $\bol{B}\lr{\bol{x},t}$
at some instant $t$, the 5 variables $\lr{\bol{u},\rho,P}$ are determined
by the 5 equations \eqref{MHD2-1}, \eqref{MHD2-2}, and \eqref{MHD2-4}, which represent a first order nonlinear system of PDEs.

In the following, we are concerned with 
non-vacuum configurations such that the electric current and the magnetic field are not collinear, i.e. $\lr{\nabla\cp\bol{B}}\cp\bol{B}\neq\bol{0}$.
By dotting equation \eqref{MHD2-1} with $\bol{u}$ and using equation \eqref{MHD2-2} we thus see that 
\begin{equation}
\bol{u}=\alpha\nabla\cp\bol{B}+\beta\bol{B},\label{u}
\end{equation}
where $\alpha\lr{\bol{x},t}$ and $\beta\lr{\bol{x},t}$ are functions determined by system \eqref{MHD2}, provided that it admits solutions.
Equation \eqref{u} implies that the dominant contributions to the flow velocity $\bol{u}$ are either in the direction of the electric current, or along the magnetic field itself.

Next, consider the ordering condition involving $\p\bol{B}/\p t=-\nabla\cp\bol{E}$ in \eqref{ordering} when the electric field is given by the generalized Ohm's law \eqref{GOL} and the velocity field has expression \eqref{u}. 
It follows that 
for all $t\geq 0$ there exists 
a single valued function $\Psi\lr{\bol{x},t}$ such that  
\begin{equation}
\frac{T^2}{L^2\sqrt{\mu_0\rho_0}}\lrs{\lr{\alpha-\frac{\kappa}{\rho}}\lr{\nabla\cp\bol{B}}\cp\bol{B}-\frac{\epsilon}{T}\nabla \Psi}\sim {\epsilon^2}
,\label{FS}
\end{equation}
where the factor $\epsilon/T$ in front of $\Psi$ emphasizes the fact that $\Psi$ 
scales as $\bol{B}\sim L^{-2}\Psi$ and that it 
has dimensions of magnetic flux, and $\rho_0\in\mathbb{R}$, $\rho_0>0$, denotes a reference mass density. 
Equation \eqref{FS} implies that the magnetic field possesses approximate flux surfaces defined by level sets of $\Psi$. 
Recalling the boundary conditions \eqref{bc}, from \eqref{FS} we must also have
\begin{equation}
\bol{n}=\pm\frac{\lr{\nabla\cp\bol{B}}\cp\bol{B}}{\abs{\lr{\nabla\cp\bol{B}}\cp\bol{B}}}=
\pm\frac{\nabla\Psi}{\abs{\nabla\Psi}}+o\lr{\epsilon}~~~~{\rm on}~~\p\Omega,
\end{equation}
where in this notation the sign $\pm$ depends on the orientation of
$\lr{\nabla\cp\bol{B}}\cp\bol{B}$ on $\p\Omega$. 
Taking the square of both sides of this equation shows that $o\lr{\epsilon}=0$ on $\p\Omega$. Hence, 
\begin{equation}
\bol{n}=\pm\frac{\nabla\Psi}{\abs{\nabla\Psi}}~~~~{\rm on}~~\p\Omega,
\end{equation}
This result also implies that $\Psi$ is constant on $\p\Omega$. 

Since the magnetic field is solenoidal, equation \eqref{FS} is sufficient to infer that 
there exists a single valued function $\Theta$ such that 
\begin{equation}
\bol{B}=\nabla\Psi\cp\nabla \Theta+\frac{L}{T}\sqrt{\mu_0\rho_0}\, o\lr{\epsilon},\label{B-Cl}
\end{equation}
in any sufficiently small neighborhood $U\subseteq\Omega$ provided that $\bol{B}$ is sufficiently regular 
 (Lie-Darboux theorem \cite{Darboux}). 
 Note that the normalization factor $L\sqrt{\mu_0\rho_0}/T$ must be kept  
 because in general both $\bol{B}$ and $\rho_0$ can be large (only the size of the ratio $T\bol{B}/L\sqrt{\mu_0\rho}$ is controlled by the ordering \eqref{ordering}). At leading order we may therefore set
 \begin{equation}
\bol{B}=\nabla\Psi\cp\nabla\Theta.\label{B-Cl2}
 \end{equation}
 The local nature of the Clebsch representation \eqref{B-Cl2} can be overcome by allowing $\Theta$ to be a multi-valued potential (angle-type variable). 
Substituting the Clebsch representation \eqref{B-Cl2} into the third equation in system \eqref{MHD2} for magnetic field induction, defining
\begin{equation}
\bol{\xi}=\bol{u}-\frac{\kappa}{\rho}\nabla\cp\bol{B},\label{q} 
\end{equation} 
and using standard vector identities leads to
\begin{equation}
\nabla\frac{\p\Psi}{\p t}\cp\nabla\Theta+\nabla\Psi\cp\nabla\frac{\p\Theta}{\p t}=\nabla\lr{\bol{\xi}\cdot\nabla\Theta}\cp{\nabla\Psi}-\nabla\lr{\bol{\xi}\cdot\nabla\Psi}\cp\nabla\Theta,
\end{equation}
This equation can be rearranged as
\begin{equation}
\nabla\lr{\frac{\p\Psi}{\p t}+\bol{\xi}\cdot\nabla\Psi}\cp\nabla\Theta=\nabla\lr{\frac{\p\Theta}{\p t}+\bol{\xi}\cdot\nabla\Theta}\cp\nabla\Psi.
\end{equation}
If $\bol{B}\neq\bol{0}$, the vector fields $\nabla\Psi$ and $\nabla\Theta$ 
are linearly independent. Therefore, by dotting this equation by $\nabla\Psi$ and $\nabla\Theta$ one finds 
\begin{subequations}
\begin{align}
\frac{\p\Psi}{\p t}+\bol{\xi}\cdot\nabla\Psi=&f\lr{\Psi,\Theta},\\
\frac{\p\Theta}{\p t}+\bol{\xi}\cdot\nabla\Theta=&g\lr{\Psi,\Theta},
\end{align}\label{MHD-3}
\end{subequations}
where $f\lr{\Psi,\Theta}$ and $g\lr{\Psi,\Theta}$ are functions of $\Psi$ and $\Theta$ 
such that $\p f/\p\Psi=-\p g/\p\Theta$.
Next, observe that the first (momentum) equation in the MHD system \eqref{MHD2} implies that equilibria without flow satisfy
\begin{equation}
\frac{\p P_0}{\p\Psi}\nabla\Psi+\frac{\p P_0}{\p\Theta}\nabla\Theta=\frac{1}{\mu_0}\lr{\nabla\cp\bol{B}}\cp\bol{B}=\frac{1}{\mu_0}\lrs{\lr{\nabla\cp\bol{B}\cdot\nabla\Theta}\nabla\Psi-\lr{\nabla\cp\bol{B}\cdot\nabla\Psi}\nabla\Theta},
\end{equation}
with $P_0=P_0\lr{\Psi,\Theta}$ the pressure at the instant $t=t_0$ in which the system is at equilibrium. 
Similarly, when $\p/\p t=0$ the induction equation \eqref{MHD2-3} implies that
\begin{equation}
A_0\lr{\nabla\cp\bol{B}}\cp\bol{B}=A_0\mu_0\nabla P_0=\nabla\Phi,\label{A0}
\end{equation}
for some function $\Phi$ and where $A_0$ is the value of the function
\begin{equation}
A={\alpha-\frac{\kappa}{\rho}},\label{A}
\end{equation}
at the instant $t=t_0$. 
Equation \eqref{A0} implies that $A_0$ is a function of $P_0$. 
Comparing these results with steady states of \eqref{MHD-3}, we find that at equilibrium  
\begin{equation}
f_0=-\mu_0A_0\lr{P_0}\frac{\p P_0}{\p\Theta},~~~~g_0=\mu_0A_0\lr{P_0}\frac{\p P_0}{\p\Psi}.
\end{equation}
In order to fulfill the constraint $\p f_0/\p\Psi=-\p g_0/\p\Theta$, we demand that  $P_0=P_0\lr{\Psi}$ so that $f_0=0$ and $g_0=\mu_0A_0\lr{P_0}dP_0/d\Psi$. 
The induction equation for the magnetic field can thus be written as 
\begin{subequations}
\begin{align}
\frac{\p\Psi}{\p t}+\bol{\xi}\cdot\nabla\Psi=&0,\label{psit}\\
\frac{\p\Theta}{\p t}+\bol{\xi}\cdot\nabla\Theta=&\mu_0A_0\frac{d P_0}{d\Psi}.
\end{align}\label{MHD4}
\end{subequations}
Note that solutions of system \eqref{MHD4} produce exact time-dependent solutions of system \eqref{MHD2} such that steady states without flow have pressure $P_0\lr{\Psi}$ and $A_0=A_0\lr{P_0}$. 
Using equation \eqref{u}, the vector field $\bol{\xi}$ can be written as 
\begin{equation}
\bol{\xi}=A\nabla\cp\bol{B}+\beta\bol{B}.
\end{equation}
Recalling the Clebsch representation \eqref{B-Cl2}, system \eqref{MHD4} can be equivalently expressed as
\begin{subequations}
\begin{align}
\frac{\p\Psi}{\p t}+A\nabla\cdot\lrs{\nabla\Psi\cp\lr{\nabla\Theta\cp\nabla\Psi}}=&0
,\label{psitf}\\
\frac{\p\Theta}{\p t}-A\nabla\cdot\lrs{\nabla\Theta\cp\lr{\nabla\Psi\cp\nabla\Theta}}=&\mu_0A_0\frac{d P_0}{d\Psi}.\label{thetatf}
\end{align}\label{MHD6}
\end{subequations}
The two equations appearing in \eqref{MHD6} can be regarded as a 
dynamical system which determines the nonlinear evolution of the magnetic field. Here, the function $A$ is evaluated through $\alpha$, $\rho$, and $P$, which are determined from the solution of the equations \eqref{MHD2-1}, \eqref{MHD2-2}, and \eqref{MHD2-4} for the variables $\lr{\bol{u},\rho,P}$. 

A simple closure of 
system \eqref{MHD6} can be obtained
through the following reasoning. 
Suppose that we are interested in knowing the evolution of the magnetic field around some equilibrum \eqref{MHDE} that we have obtained at some instant $t=t_0$, for example, within a stellarator. 
At $t=t_0$ the electric field is irrotational since $\bol{0}=\p_t\bol{B}=-\nabla\cp\bol{E}$. Recalling \eqref{GOL} 
we have \eqref{A0} so that $A$ is a function of $P$.
Furthermore, we may identify the flux function with the equilibrium pressure field, $P_0=\lambda\Psi$ with $\lambda\in\mathbb{R}$ a constant bearing units of pressure over magnetic flux, without loss of generality. 
When the system is perturbed at some $t=t_1$, we may consider a regime in which the fields $\bol{u}$, $\rho$, and $P$ react passively to changes in the magnetic field so that  
the functional form of $\alpha$, $\rho$, and $P$ is preserved for $t\geq t_1$. 
This amounts to assuming relations of the type
\begin{equation}
A=A_0\lr{P_0}=A_0\lr{\Psi}~~~~\forall t\geq t_0.
\end{equation}
Then, system \eqref{MHD6} reduces to an independent nonlinear system of two coupled PDEs for the variables $\Psi$ and $\Theta$, 
\begin{subequations}
\begin{align}
\frac{\p\Psi}{\p t}+A_0\lr{\Psi}\nabla\cdot\lrs{\nabla\Psi\cp\lr{\nabla\Theta\cp\nabla\Psi}}=&0,\label{psitf2}\\
\frac{\p\Theta}{\p t}-A_0\lr{\Psi}{\nabla\cdot\lrs{\nabla\Theta\cp\lr{\nabla\Psi\cp\nabla\Theta}}}=&\mu_0\lambda A_0\lr{\Psi}.\label{thetatf2}
\end{align}\label{MHD7}
\end{subequations}
A possible choice of boundary conditions for this closed dynamical system is
\begin{equation}
\Psi={\rm constant},~~~~\nabla\Theta\cdot\bol{n}=0~~~~{\rm on}~~\p\Omega.\label{bcpsith}
\end{equation}
Notice that while these boundary conditions are compatible with $\bol{B}\cdot\bol{n}=0$ and $P={\rm constant}$ on $\p\Omega$, they do not necessarily imply $\nabla\cp\bol{B}\cdot\bol{n}=0$ on $\p\Omega$. 
We also observe that since in general $\Theta$ is a multi-valued (angle) variable,
for computational purposes it may be convenient to perform a change of variables. 
For example, if $\Omega$ is a toroidal volume spanned by coordinates $\lr{\zeta,\mu,\nu}$ where level sets of $\zeta$ define nested toroidal surfaces within $\Omega$ and $\mu,\nu$ are toroidal and poloidal angles, 
one may set $\Theta=M\mu+N\nu+\chi$ with $M$ and $N$ integers, 
and consider system \eqref{MHD7} in terms of $\Psi$ and $\chi$, where $\chi\lr{\bol{x},t}$ is a single-valued function satisfying $\nabla\chi\cdot\bol{n}=-\lr{M\nabla\mu+N\nabla\nu}\cdot\bol{n}$ on $\p\Omega$.

In the following, we shall focus our attention on the derived dynamical systems \eqref{MHD2} and  \eqref{MHD6} and, in particular, on the model system \eqref{MHD7}. 
Finally, we observe that while the ordering \eqref{ordering} was considered 
for its relevance in stellarator applications, the same governing equations 
\eqref{MHD2}, \eqref{MHD6}, and \eqref{MHD7} can be obtained under more general orderings. For example, 
\begin{subequations}
\begin{align}
&{\frac{T}{L\sqrt{\mu_0\rho}}}\bol{B}\sim \frac{T}{L}\bol{u}\sim
T\nabla\cp\bol{u}\sim \frac{T}{\sqrt{\mu_0\rho}}\nabla\cp\bol{B}\sim \frac{T^2}{\rho 
 L^2}P\sim 
1,\\
& \frac{T}{\abs{\bol{u}}}\frac{\p\bol{u}}{\p t}\sim  \frac{T}{\rho}\frac{\p\rho}{\p t}\sim \frac{T}{\abs{\bol{B}}}\frac{\p\bol{B}}{\p t}\sim\frac{T}{\bol{\abs{B}}}\nabla\cp\bol{E}\sim \frac{T^3}{L^2}\bol{u}\cdot\lr{\frac{1}{\rho}\nabla P+\frac{1}{2}\nabla\bol{u}^2}\sim\epsilon.
\end{align}\label{ordering2}
\end{subequations}

\section{Conservation laws and relaxed states}
In this section we first discuss the invariants of systems \eqref{MHD2}, \eqref{MHD6}, and \eqref{MHD7}. 
Then, these invariants are used to construct a variational principle describing steady configurations of system \eqref{MHD7}. These steady states correspond to MHD equilibria \eqref{MHDE} and  can be understood as the result of a constrained relaxation process in which the weakest invariant is dissipated while the others  are kept constant. 

\subsection{Conservation of magnetic energy, helicity, and flux}

System \eqref{MHD2} is endowed with invariants.
To see this, first observe that the induction equation therein can be written as
\begin{equation}
\frac{\p\bol{B}}{\p t}=\nabla\cp\lrs{A\lr{\nabla\cp\bol{B}}\cp\bol{B}},\label{ind}
\end{equation}
where we used \eqref{u} and \eqref{A}. 
The magnetic energy of the system is given by
\begin{equation}
M_{\Omega}=\frac{1}{2\mu_0}\int_{\Omega}\bol{B}^2\,dV.\label{MOm}
\end{equation}
From \eqref{ind}, it follows that
\begin{equation}
\frac{dM_{\Omega}}{dt}=\frac{1}{\mu_0}\int_{\p\Omega}A\lrs{\lr{\nabla\cp\bol{B}}\cp\bol{B}}\cp\bol{B}\cdot\bol{n}\,dS=\frac{1}{\mu_0}\int_{\Omega}A\lrs{\lr{\bol{B}\cdot\nabla\cp\bol{B}}\bol{B}-\bol{B}^2\nabla\cp\bol{B}}\cdot\bol{n}\,dS=0,
\end{equation}
where we have used the boundary conditions \eqref{bc}.
This shows that the magnetic energy $M_{\Omega}$ is an invariant of system \eqref{MHD2}. 
Let us verify that $M_{\Omega}$ is also an invariant of the reduced systems \eqref{MHD6} and \eqref{MHD7} under the first  boundary condition in \eqref{bcpsith}, i.e. $\Psi={\rm constant}$ on $\p\Omega$. 
Noting that $\p\Psi/\p t=-A\nabla\cp\bol{B}\cdot\nabla\Psi=0$ as well as $\bol{n}\cp\nabla\Psi=\bol{0}$ on $\p\Omega$, and that $\p\Theta/\p t$ is single-valued, we have 
\begin{equation}
\begin{split}
\frac{dM_{\Omega}}{dt}=&\frac{1}{\mu_0}\int_{\Omega}\lr{\frac{\p\Psi}{\p t}\nabla\cp\bol{B}\cdot\nabla\Theta-\frac{\p\Theta}{\p t}\nabla\cp\bol{B}\cdot\nabla\Psi}dV=-\lambda\int_{\Omega} A_0\lr{\Psi}\nabla\cp\bol{B}\cdot\nabla\Psi\,dV\\=&-\lambda\int_{\p\Omega}A_0\lr{\Psi}\bol{B}\cp\nabla\Psi\cdot\bol{n}\,dS=0. 
\end{split}
\end{equation}

Next, 
consider the magnetic helicity 
\begin{equation}
K_{\Omega}=\frac{1}{2}\int_{\Omega}\bol{A}\cdot\bol{B}\,dV,
\end{equation}
where $\bol{A}\lr{\bol{x},t}$ is
a single-valued vector potential such that $\bol{B}=\nabla\cp\bol{A}$. 
We have
\begin{equation}
\begin{split}
\frac{dK_{\Omega}}{dt}=&\int_{\p\Omega}\lrc{\frac{1}{2}\bol{A}\cp\frac{\p\bol{A}}{\p t}+A\lrs{\lr{\nabla\cp\bol{B}}\cp\bol{B}}\cp\bol{A}}\cdot\bol{n}\,dS\\=&\int_{\p\Omega}\lrc{\frac{1}{2}\bol{A}\cp\frac{\p\bol{A}}{\p t}+A\lrs{\lr{\bol{A}\cdot\nabla\cp\bol{B}}\bol{B}-\lr{\bol{A}\cdot\bol{B}}\nabla\cp\bol{B}}}\cdot\bol{n}\,dS=\frac{1}{2}\int_{\p\Omega}\bol{A}\cp\frac{\p\bol{A}}{\p t}\cdot\bol{n}\,dS.
\end{split}
\end{equation}
On the other hand, the induction equation \eqref{ind} implies that
\begin{equation}
\frac{\p\bol{A}}{\p t}=\frac{\p\bol{q}}{\p t}+A\lr{\nabla\cp\bol{B}}\cp\bol{B},
\end{equation}
where $\bol{q}\lr{\bol{x},t}\in{\rm ker}\lr{\rm curl}$, i.e. $\nabla\cp\bol{q}=\bol{0}$.  
Due to gauge freedom, the vector field $\bol{q}$ can be absorbed 
in the definition of $\bol{A}$ without loss of generality so that at the boundary we have
\begin{equation}
\frac{\p\bol{A}}{\p t}\cp\bol{n}=A\lrs{\lr{\nabla\cp\bol{B}}\cp\bol{B}}\cp\bol{n}=\bol{0}~~~~{\rm on}~~\p\Omega, \end{equation}
where we used the boundary conditions \eqref{bc}.  
Hence, the magnetic helicity $K_{\Omega}$ is an invariant of system \eqref{MHD2}. 

The magnetic helicity $K_{\Omega}$ is also an  invariant of \eqref{MHD6} and \eqref{MHD7} under the boundary condition $\Psi={\rm constant}$ on $\p\Omega$. However, it degenerates to a trivial invariant $K_{\Omega}=0$ when $\p\Omega$ is connected \cite{Pfeff}. 
To see this, consider the single-valued vector potential  
\begin{equation}
\bol{A}=\bol{q}_0\lr{\bol{x}}+\Psi\nabla\Theta,
\end{equation}
with $\bol{q}_0\lr{\bol{x}}\in {\rm ker}\lr{{\rm curl}}$ a time-independent  gauge vector field such that $\nabla\cp\bol{q}_0=0$. 
We have
\begin{equation}
K_{\Omega}=\frac{1}{2}\int_{\Omega}\bol{q}_0\cdot\nabla\Psi\cp\nabla\Theta\,dV=\frac{1}{2}\int_{\p\Omega}\Psi\nabla\Theta\cp\bol{q}_0\cdot\bol{n}\,dS,\label{K2}
\end{equation}
as well as
\begin{equation}
\frac{dK_{\Omega}}{dt}=\frac{1}{2}\int_{\Omega}\bol{q}_0\cdot\lr{\nabla\frac{\p\Psi}{\p t}\cp\nabla\Theta+\nabla\Psi\cp\nabla\frac{\p\Theta}{\p t}}dV=\frac{1}{2}\int_{\p\Omega}\lr{\frac{\p\Psi}{\p t}\nabla\Theta\cp\bol{q}_0+\frac{\p\Theta}{\p t}\bol{q}_0\cp\nabla\Psi}\cdot\bol{n}\,dS=0,
\end{equation}
where in the last passage we used the fact that $\p\Psi/\p t=-A\nabla\cp\bol{B}\cdot\nabla\Psi=0$ on $\p\Omega$ due to the boundary conditions \eqref{bcpsith}, the fact that $\p\Theta/\p t$ is single-valued, and the fact that $\bol{n}\cp\nabla\Psi=\bol{0}$ on $\p\Omega$. 
This shows that 
the magnetic helicity $K_{\Omega}$ 
defined in equation \eqref{K2} is an invariant of both systems \eqref{MHD6} 
and \eqref{MHD7}. 
However, if the boundary $\p\Omega$ is a connected surface (e.g. the boundary of a solid toroidal volume), equation \eqref{K2} can be written as
\begin{equation}
K_{\Omega}=\frac{1}{2}\Psi_e\int_{\Omega}\nabla\cdot\lr{\nabla\Theta\cp\bol{q}_0}\,dV=0,
\end{equation}
where $\Psi\vert_{\p\Omega}=\Psi_e\in\mathbb{R}$ is the boundary value of $\Psi$. 
This shows that the magnetic helicity $K_{\Omega}$ becomes a trivial invariant whenever $\p\Omega$ defines a connected surface.   
Nevertheless, $K_{\Omega}$ is nontrivial when $\p\Omega$ is not a connected surface (such as when $\p\Omega$ is the boundary of a hollow toroidal volume).




Now consider the functional 
\begin{equation}
F_{\Omega}=\int_{\Omega} f\lr{\Psi}\,dV,
\end{equation}
where $f$ is any function of $\Psi$. 
When $f=\Psi$, $F_{\Omega}$ is the total magnetic flux
\begin{equation}
\mc{F}_{\Omega}=\int_{\Omega}\Psi\,dV. 
\end{equation}
On the other hand, recall that system \eqref{MHD6} corresponds to the induction equation of system \eqref{MHD2} under the assumption \eqref{FS} regarding the existence of approximate flux surfaces. 
Therefore, the rate of change of $F_{\Omega}$ following from system \eqref{MHD6} is given by 
\begin{equation}
\begin{split}
\frac{dF_{\Omega}}{dt}&=\int_{\Omega}\frac{d f}{d\Psi}\lrc{-A\nabla\cdot\lrs{\nabla\Psi\cp\lr{\nabla\Theta\cp\nabla\Psi}}
}dV=-\int_{\Omega}\lrs{\nabla A\cdot\nabla f\cp\lr{\nabla\Psi\cp\nabla\Theta}
}dV.
\end{split}
\end{equation}
When $A=A_0\lr{\Psi}$ 
this integral identically vanishes. 
This shows that 
the functional $F_{\Omega}$ is an invariant of \eqref{MHD7}. 
It should be emphasized that the conservation of $F_{\Omega}$ 
is, in general, a property that is favorable to plasma confinement.  
Indeed, if we regard the density $\rho\approx\rho\lr{\Psi}$ as 
a function of the magnetic flux, 
conservation of total particle number  $\int_{\Omega}\rho\lr{\Psi}dV$ implies, for example, that large migrations of particles from regions of large $\Psi$ to regions of low $\Psi$ cannot occur without breaking the constancy of the density weighted magnetic flux $\int_{\Omega}\rho\lr{\Psi}\Psi\,dV$. 

The invariants of systems \eqref{MHD2} and  \eqref{MHD7} are summarized in table 1 and 2 respectively. 
\begin{table}
\begin{center}
\begin{tabular}{c c c c c c} 
 \hline
 \hline
 Invariant & Expression & Field Conditions & Boundary Conditions\\  
 \hline 
 Magnetic energy $M_{\Omega}$ & $\frac{1}{2\mu_0}\int_{\Omega}\bol{B}^2dV$ &  &  $\bol{B}\cdot\bol{n}=0$, $\nabla\cp\bol{B}\cdot\bol{n}=0$ &\\
 Magnetic helicity $K_{\Omega}$ & $\frac{1}{2}\int_{\Omega}\bol{A}\cdot\bol{B}\,dV$ & $\frac{\p\bol{q}}{\p t}=\frac{\p\bol{A}}{\p t}-A\lr{\nabla\cp\bol{B}}\cp\bol{B}=\bol{0}$ &  $\bol{B}\cdot\bol{n}=0$, $\nabla\cp\bol{B}\cdot\bol{n}=0$ &\\
 \hline
 \hline
\end{tabular}
\caption{\label{tab2} Invariants of system \eqref{MHD2}. Field conditions for the conservation of magnetic helicity $K_{\Omega}$ specify the gauge $\p\bol{q}/\p t$ of the vector potential $\bol{A}$.}
\end{center}
\end{table}
\begin{table}
\begin{center}
\begin{tabular}{c c c c c c} 
 \hline
 \hline
 Invariant & Expression & Field Conditions & Boundary Conditions\\  
 \hline 
 Magnetic energy $M_{\Omega}$ & $\frac{1}{2\mu_0}\int_{\Omega}\bol{B}^2dV$ & $\bol{B}=\nabla\Psi\cp\nabla\Theta$ &  $\Psi={\rm constant}$ &\\
Magnetic helicity $K_{\Omega}$ & $\frac{1}{2}\int_{\Omega}\bol{A}\cdot\bol{B}\,dV$ & $\bol{B}=\nabla\Psi\cp\nabla\Theta$, $\bol{A}=\bol{q}_0+\Psi\nabla\Theta$ & 
\begin{tabular}{@{}c@{}} 
$\Psi={\rm constant}$,  
\\ $\p\Omega$ not connected
\end{tabular} \\
 Magnetic flux $F_{\Omega}$ & $\int_{\Omega}f\lr{\Psi}dV$ & $\bol{B}=\nabla\Psi\cp\nabla\Theta$ &  $\Psi={\rm constant}$ &\\
 \hline
 \hline
\end{tabular}
\caption{\label{tab2} Invariants of system \eqref{MHD7}. The magnetic helicity $K_{\Omega}$ degenerates to a trivial invariant $K_{\Omega}=0$ when $\p\Omega$ is a connected surface.}
\end{center}
\end{table}

\subsection{Steady states and relaxation}

In the remaining part of this section we restrict our attention to the variational formulation of steady states associated with 
the model system \eqref{MHD7}. 
From a physical standpoint one expects equilibrium states to correspond to critical points of an energy functional.  
In the present setting, the energy involved is that  associated with the magnetic field $\bol{B}=\nabla\Psi\cp\nabla\Theta$. 
Here, we consider the target functional
\begin{equation}
W_{\Omega}=M_{\Omega}-\lambda {F}_{\Omega}=\int_{\Omega}\lrs{\frac{\bol{B}^2}{2\mu_0}- \lambda f\lr{\Psi}}dV=\int_{\Omega}\lrs{\frac{1}{2\mu_0}\abs{\nabla\Psi\cp\nabla\Theta}^2-\lambda f\lr{\Psi}}dV, 
\end{equation}
Note that $W_{\Omega}$ comprises the magnetic energy $M_{\Omega}$ and the functional of the magnetic flux ${F}_{\Omega}$, which are invariants of system \eqref{MHD7}.
In practice, non-ideal processes involving dissipation result in 
faster violation of ideal invariants that include higher order derivatives of the dynamical variables \cite{YosMahPRL}. Hence, we expect equilibrium states to be the result of the minimization by dissipation of $M_{\Omega}$ under the
constraint of preserved magnetic  flux ${F}_{\Omega}$. 
The constant $\lambda$ thus plays the role of a Langrange multiplier. 
The relaxation scenario described above is analogous to so-called Taylor relaxation in which a Beltrami state is produced as a result of a dissipation process in which magnetic energy is minimized under the constraint of magnetic helicity \cite{Taylor}. 
Notice also that Taylor relaxation can be expected in the context of system \eqref{MHD2}, 
whose invariants include the magnetic energy $M_{\Omega}$ and  the magnetic helicity $K_{\Omega}$.

One can verify \cite{Grad58,SY23} that setting to zero the first variation of the functional $W_{\Omega}$ with respect to $\Psi$ and $\Theta$ under the assumption that $\delta \Psi$ and $\delta\Theta$ identically vanish on the boundary, or $\delta\Psi=0$ and $\Psi={\rm constant}$ on the boundary,  gives the 
system of equations
\begin{subequations}
\begin{align}
\nabla\cdot\lrs{\nabla\Theta\cp\lr{\nabla \Psi\cp\nabla\Theta}}&=-\lambda\mu_0\frac{d f}{d\Psi},\\
\nabla\cdot\lrs{\nabla \Psi\cp\lr{\nabla \Theta\cp\nabla \Psi}}&=0,
\end{align}\label{MHS2}
\end{subequations}
which are equivalent to the MHD equilibrium equations \eqref{MHDE} with pressure $P=\lambda f$ upon substitution of $\bol{B}=\nabla \Psi\cp\nabla\Theta$. 
Notice also that solutions of \eqref{MHS2} give steady solutions of \eqref{MHD7} 
when the choice $f=\Psi$ 
is made.


\section{Hamiltonian structure}
The aim of this section is to show that the model system \eqref{MHD7} is endowed with a Hamiltonian structure.
This property will later be used to discuss the nonlinear stability of steady solutions. 

Let $\delta F/\delta \Psi$ denote the functional derivative of the functional $F:\mf{X}\rightarrow \mathbb{R}$ on the state space $\mf{X}$ with respect to $\Psi\in\mf{X}$. We have
\begin{proposition}
System \eqref{MHD7} is a Hamiltonian system with Poisson bracket 
\begin{equation}
\lrc{F,G}=\mu_0\int_{\Omega}A_0\lr{\Psi}\lr{\frac{\delta F}{\delta\Psi}\frac{\delta G}{\delta\Theta}-\frac{\delta F}{\delta\Theta}\frac{\delta G}{\delta\Psi}}dV,\label{PB}
\end{equation}
and Hamiltonian 
\begin{equation}
H_{\Omega}=\int_{\Omega}\lr{\frac{1}{2\mu_0}\abs{\nabla\Psi\cp\nabla\Theta}^2-\lambda\Psi}dV.\label{HOm}
\end{equation}
\end{proposition}
\noindent To verify proposition 1, first recall that in a Hamiltonian system the evolution of a physical observable $F\in\mf{X}^{\ast}$ satisfies the equation of motion 
\begin{equation}
\frac{\p F}{\p t}=\lrc{F,H_{\Omega}},
\end{equation}
where $\mf{X}$ is a vector space, $\mf{X}^{\ast}$ the set of differentiable functionals $F:\mf{X}\rightarrow\mathbb{R}$, and $H_{\Omega}\in\mf{X}^{\ast}$ the Hamiltonian, and the Poisson bracket $\lrc{\circ,\circ}:\mf{X}^{\ast}\cp\mf{X}^{\ast}\rightarrow\mf{X}^{\ast}$ satisfies the axioms of bilinearity, alternativity, antisymmetry, Leibniz rule, and Jacobi identity, 
\begin{subequations}
\begin{align}
&\lrc{a F+b G,H}=a\lrc{F,H}+b\lrc{G,H},~~~~\lrc{F,a G+bH}=a\lrc{F,G}+b\lrc{F,H},\label{bil}\\
&\lrc{F,F}=0,\label{alt}\\
&\lrc{F,G}=-\lrc{G,F},\label{asym}\\
&\lrc{FG,H}=\lrc{F,H}G+F\lrc{G,H},\label{Leib}\\
&\lrc{F,\lrc{G,H}}+\lrc{G,\lrc{H,F}}+\lrc{H,\lrc{F,G}}=0,\label{JI}
\end{align}\label{axioms}
\end{subequations}
$\forall a,b\in\mathbb{R}$ and $F,G,H\in\mf{X}^{\ast}$.
For completeness, we remark that antisymmetry \eqref{asym} follows from 
bilinearity \eqref{bil} and alternativity \eqref{alt} by evaluation of $\lrc{F+G,F+G}$.
The Leibniz rule \eqref{Leib} ensures that the Poisson bracket acts as a differential operator, while the Jacobi identity \eqref{JI} assigns the Lie-algebra structure. 

In the present setting, 
the state space $\mf{X}$ denotes the function space to which the dynamical variables $\Psi$ and $\Theta$ of system \eqref{MHD7} belong, while $\mf{X}^{\ast}$ represents the vector space of differentiable functionals of $\Psi$ and $\Theta$. 
Below, we show that system \eqref{MHD7} can be cast in the noncanonical Hamiltonian form
\begin{subequations}
\begin{align}
\frac{\p\Psi}{\p t}&=\lrc{\Psi,H_{\Omega}}=\mu_0A_0\lr{\Psi}\frac{\delta H_{\Omega}}{\delta\Theta},\label{psitH}\\
\frac{\p\Theta}{\p t}&=\lrc{\Theta,H_{\Omega}}=-\mu_0A_0\lr{\Psi}\frac{\delta H_{\Omega}}{\delta\Psi}.\label{thetatH}
\end{align}\label{EoMH}
\end{subequations}
To see this, let us first verify that the bracket \eqref{PB} 
correctly generates system \eqref{MHD7}. We have
\begin{equation}
\frac{\p\Psi}{\p t}=\lrc{\Psi,H_{\Omega}}=\mu_0\int_{\Omega}A_0\delta\lr{\bol{x}-\bol{x}'}\frac{\delta H_{\Omega}}{\delta\Theta}\,dV=-A_0\nabla\cdot\lrs{\nabla\Psi\cp\lr{\nabla\Theta\cp\nabla\Psi}},
\end{equation}
where we have assumed that variations $\delta\Theta$ vanish on the boundary $\p\Omega$ when evaluating $\delta H_{\Omega}/\delta\Theta$. 
The same result can be obtained by setting $\Psi={\rm constant}$ on $\p\Omega$ instead of $\delta\Theta=0$ on $\p\Omega$. 
Similarly, 
\begin{equation}
\frac{\p\Theta}{\p t}=\lrc{\Theta,H_{\Omega}}=-\mu_0\int_{\Omega}A_0\delta\lr{\bol{x}-\bol{x}'}\frac{\delta H_{\Omega}}{\delta\Psi}\,dV=A_0\nabla\cdot\lrs{\nabla\Psi\cp\lr{\nabla\Theta\cp\nabla\Psi}}+\mu_0 \lambda A_0,
\end{equation}
where we have assumed that variations $\delta\Psi$ vanish on the boundary $\p\Omega$ when evaluating $\delta H_{\Omega}/\delta\Psi$. 
This shows that the bracket \eqref{PB} generates system \eqref{MHD7} according to \eqref{EoMH}.

We are now left with the task of verifying that the bracket \eqref{PB} satisfies
the Poisson bracket axioms \eqref{axioms}. 
The verification of bilinearity \eqref{bil}, alternativity \eqref{alt}, antisymmetry \eqref{asym}, and Leibniz rule \eqref{Leib}  
is immediate.
Denoting with $\circlearrowright$ summation of even permutations,
and introducing the simplified notation $F_{\Psi}=\delta F/\delta\Psi$ for functional derivatives, the Jacobi identity \eqref{JI} can be evaluated as 
\begin{equation}
\begin{split}
\lrc{F,\lrc{G,H}}+\circlearrowright=&
\mu_0^2\int_{\Omega}A_0\lrc{F_{\Psi}\frac{\delta}{\delta\Theta}\lrs{\int_{\Omega}A_0\lr{G_{\Psi}H_{\Theta}-G_{\Theta}H_{\Psi}}dV}}dV\\&-\mu_0^2\int_{\Omega}A_0\lrc{F_{\Theta}\frac{\delta}{\delta\Psi}\lrs{\int_{\Omega}A_0\lr{G_{\Psi}H_{\Theta}-G_{\Theta}H_{\Psi}}dV}}dV+\circlearrowright\\
=&-\mu_0^2\int_{\Omega}A_0\frac{dA_0}{d\Psi}F_{\Theta}\lr{G_{\Psi}H_{\Theta}-G_{\Theta}H_{\Psi}}dV+\circlearrowright\\
=&0,
\end{split}\label{JI2}
\end{equation}
where we used the fact that terms involving second order functional derivatives \cite{Morrison89,Little89,Olver} of 
$F$, $G$, and $H$ and terms containing $dA_0/d\Psi$ cancel upon summation of even permutations. 
Equation \eqref{JI2} shows that the bracket \eqref{PB} also satisfies the Jacobi identity, and thus it is a Poisson bracket.



\section{Remarks on the nonlinear stability of steady solutions}

Let $\chi=\lr{\Psi,\Theta}\in\mf{X}$ denote a point in the space of solutions $\mf{X}$ of system \eqref{MHD7}.
The Hamiltonian nature of system \eqref{MHD7} implies that 
critical points of the Hamiltonian $H_{\Omega}$, i.e. points $\chi_0=\lr{\Psi_0,\Theta_0}\in\mf{X}$ such that the first variation of $H_{\Omega}$ vanishes, 
\begin{equation}
\delta H_{\Omega}\lrs{\chi_0}=0,
\end{equation}
correspond to steady states of system \eqref{MHD7}. Indeed, 
at $\chi_0$ both $\delta H_{\Omega}/\delta\Psi$ and $\delta H_{\Omega}/\delta\Theta$ vanish, implying that $\p\Psi/\p t=\p\Theta/\p t=0$ 
(recall equation \eqref{EoMH}). 
Information on the stability of a 
critical point $\chi_0$ can be gained with the aid of conservation of energy  if norms $\Norm{\cdot}_1:\mf{X}\rightarrow\mathbb{R}$ and $\Norm{\cdot}_2:\mf{X}\rightarrow\mathbb{R}$ and positive real constants $\mc{C},\mc{C}'$ can be found such that 
\begin{equation}
\mc{C}\Norm{\chi\lr{t}-\chi_0}_1^2\leq\abs{H_{\Omega}\lrs{\chi\lr{t}}-H_{\Omega}\lrs{\chi_0}}=\abs{H_{\Omega}\lrs{\chi\lr{0}}-H_{\Omega}\lrs{\chi_0}}\leq \mc{C}'\Norm{\chi\lr{0}-\chi_0}_2^2~~~~\forall t\geq 0.\label{Nstab}
\end{equation}
Here, the simplified notation $\chi\lr{t}=\chi\lr{\bol{x},t}$ has been used. 
Note that if equation \eqref{Nstab} holds, then a solution $\chi$ of system \eqref{MHD7} that initially differs from the critical point $\chi_0$ by the amount $\delta\chi_0=\chi\lr{0}-\chi_0$ 
remains close to $\chi_0$ at all later times $t\geq 0$ as prescribed by the norms $\Norm{\cdot}_1$ and $\Norm{\cdot}_2$. Usually, $\Norm{\cdot}_1$ and $\Norm{\cdot}_2$ satisfy the norm axioms, and assign a topology to the state space $\mf{X}$.  
However, this requirement can be relaxed by demanding $\Norm{\cdot}_1$ and $\Norm{\cdot}_2$ to be positive real valued functions that provide some measure of the 
distance between two points in some subspace of the state space $\mf{X}$. Furthermore, it is common to have a single distance function  $\Norm{\cdot}=\Norm{\cdot}_1=\Norm{\cdot}_2$. 
Using the expression \eqref{HOm} for the Hamiltonian $H_{\Omega}$ and the boundary conditions $\delta\Psi=\delta\Theta=0$ on $\p\Omega$ or $\delta\Psi=0$ and $\Psi={\rm constant}$ on $\p\Omega$, one finds that
\begin{equation}
\begin{split}
\mu_0H_{\Omega}&\left[\Psi+\delta\Psi,\Theta+\delta\Theta\right]-\mu_0H_{\Omega}\left[\Psi,\Theta\right]=\\&
-\int_{\Omega}\lrc{\delta\Theta\nabla\cdot\lrs{\nabla\Psi\cp\lr{\nabla\Theta\cp\nabla\Psi}}+\delta\Psi\nabla\cdot\lrs{\nabla\Theta\cp\lr{\nabla\Psi\cp\nabla\Theta}}+\mu_0\lambda\delta\Psi}dV\\
&+\int_{\Omega}\lrs{\frac{1}{2}\lr{\nabla\Psi\cp\nabla\delta\Theta+\nabla\delta\Psi\cp\nabla\Theta}^2+\nabla\Psi\cp\nabla\Theta\cdot\nabla\delta\Psi\cp\nabla\delta\Theta}dV\\
&+\int_{\Omega}\lr{\nabla\Psi\cp\nabla\delta\Theta\cdot\nabla\delta\Psi\cp\nabla\delta\Theta+\nabla\delta\Psi\cp\nabla\Theta\cdot\nabla\delta\Psi\cp\nabla\delta\Theta}dV\\
&+\frac{1}{2}\int_{\Omega}\abs{\nabla\delta\Psi\cp\nabla\delta\Theta}^2dV.\label{dH}
\end{split}
\end{equation}
Due to the presence of second order and third order terms in the variations $\delta\Psi$ and $\delta\Theta$ without a definite sign, the difference \eqref{dH} cannot be used to define  
the functions $\Norm{\cdot}_1$ and $\Norm{\cdot}_2$ in general. 
However, the situation is different if we consider perturbations that involve only one of the Clebsch potentials $\Psi$ and $\Theta$, i.e. if either  $\delta\Psi$ or $\delta\Theta$ vanishes. 
In particular, we have the following

\begin{proposition}
Critical points $\chi_0=\lr{\Psi_0,\Theta_0}$ of system \eqref{MHD7} are nonlinearly stable against perturbations of $\Psi_0$ in the 
distance $\Norm{\Psi}^2_{\Theta}=\frac{1}{2}\int_{\Omega}\abs{\nabla\Psi\cp\nabla\Theta}^2dV$.  
In particular, for all $t\geq 0$  
\begin{equation}
\begin{split}
\Norm{\Psi\lr{t}-\Psi_0}^2_{\Theta_0}=&\frac{1}{2\mu_0}\int_{\Omega}\abs{\nabla\lr{\Psi\lr{t}-\Psi_0}\cp\nabla\Theta_0}^2dV=\abs{H_{\Omega}\lrs{\Psi\lr{t},\Theta_0}-H_{\Omega}\lrs{\Psi_0,\Theta_0}}\\=&\abs{H_{\Omega}\lrs{\Psi\lr{0},\Theta_0}-H_{\Omega}\lrs{\Psi_0,\Theta_0}}=\frac{1}{2\mu_0}\int_{\Omega}\abs{\nabla\lr{\Psi\lr{0}-\Psi_0}\cp\nabla\Theta_0}^2dV=\Norm{\Psi\lr{0}-\Psi_0}^2_{\Theta_0}.
\end{split}
\end{equation}
\end{proposition} 
\noindent The proof of proposition 2 
follows by evaluating the difference \eqref{dH} for $\delta\Theta=0$. 
Similarly, 
\begin{proposition}
Critical points $\chi_0=\lr{\Psi_0,\Theta_0}$ of system \eqref{MHD7} are nonlinearly stable against perturbations of $\Theta_0$ in the 
distance $\Norm{\Theta}^2_{\Psi}=\frac{1}{2}\int_{\Omega}\abs{\nabla\Psi\cp\nabla\Theta}^2dV$.  
In particular, for all $t\geq 0$  
\begin{equation}
\begin{split}
\Norm{\Theta\lr{t}-\Theta_0}^2_{\Psi_0}=&\frac{1}{2\mu_0}\int_{\Omega}\abs{\nabla{\Psi_0}\cp\nabla\lr{\Theta\lr{t}-\Theta_0}}^2dV=\abs{H_{\Omega}\lrs{\Psi_0,\Theta\lr{t}}-H_{\Omega}\lrs{\Psi_0,\Theta_0}}\\=&\abs{H_{\Omega}\lrs{\Psi_0,\Theta\lr{0}}-H_{\Omega}\lrs{\Psi_0,\Theta_0}}=\frac{1}{2\mu_0}\int_{\Omega}\abs{\nabla{\Psi_0}\cp\nabla\lr{\Theta\lr{0}-\Theta_0}}^2dV=\Norm{\Theta\lr{0}-\Theta_0}^2_{\Psi_0}.
\end{split}
\end{equation}
\end{proposition} 
\noindent Again, the proof of proposition 3 follows by evaluation of the difference \eqref{dH} for $\delta\Psi=0$. 

We conclude this section by observing that the distance functions $\Norm{\Psi}^2_{\Theta}$ and $\Norm{\Theta}^2_{\Psi}$ behave as seminorms due to the degeneracy brought by the cross product. 
For example, the distance function $\Norm{\Psi}^2_{\Theta}$ is degenerate since 
$\Norm{\Psi}^2_{\Theta}=\Norm{\Psi+g\lr{\Theta}}^2_{\Theta}$ for any function $g\lr{\Theta}$. 
Such degeneracy can be removed by restricting the state space $\mf{X}$ 
to include only those functions $\Psi$ such that $\int_{\Sigma_{\Theta}}\Psi\,dS=0$ 
where $\Sigma_{\Theta}=\lrc{\bol{x}\in\Omega:\Theta\lr{\bol{x}}=c\in\mathbb{R}}$ is a level set of $\Theta$ with surface element $dS$. Then, $\Norm{\Psi}_{\Theta}^2=0$ 
if and only if $\Psi=g\lr{\Theta}$. On the other hand, $\int_{\Sigma_{\Theta}}\Psi\,dS=g\lr{\Theta}\Sigma_{\Theta}=0$ if and only if $g=0$. 
Alternatively, one may simply interpret propositions 2 and 3
as constraints on the size of $\nabla\lr{\Psi\lr{t}-\Psi_0}$ and $\nabla\lr{\Theta\lr{t}-\Theta_0}$ across $\nabla\Theta_0$ and $\nabla\Psi_0$ respectively (the gradients of the variations are constrained on the submanifolds defined by level sets of $\Theta_0$ and $\Psi_0$). 


\section{Construction of MHD equilibria by double bracket dissipation}
In this section we are concerned with the following system of two nonlinear PDEs in $\Omega$, 
\begin{subequations}
\begin{align}
\frac{\p\Psi}{\p t}=&\gamma A_0\nabla\cdot\lrs{\nabla\Theta\cp\lr{\nabla\Psi\cp\nabla\Theta}}+\gamma \mu_0A_0\lambda,\\
\frac{\p\Theta}{\p t}=&\sigma A_0\nabla\cdot\lrs{\nabla\Psi\cp\lr{\nabla\Theta\cp\nabla\Psi}},
\end{align}\label{dbMHD}
\end{subequations}
subject to $A_0\lr{\Psi}>0$ and the boundary conditions \eqref{bcpsith} and 
where the positive constants $\gamma,\sigma$ bear physical units. 
System \eqref{dbMHD} has the following properties:
\begin{proposition}
Steady states of system \eqref{dbMHD} correspond to MHD equilibria  
\begin{equation}
\lr{\nabla\cp\bol{B}}\cp\bol{B}=\mu_0\lambda\nabla\Psi,~~~~\nabla\cdot\bol{B}=0,\label{MHSdiss}
\end{equation}
with $\bol{B}=\nabla\Psi\cp\nabla\Theta$. 
Furthermore, the energy $H_{\Omega}$ is progressively dissipated 
\begin{equation}
\frac{dH_{\Omega}}{dt}\leq 0~~~~\forall t\geq 0.
\end{equation}
\end{proposition}
\noindent Hence, if solutions $\lr{\Psi,\Theta}$ exist in the limit $t\rightarrow +\infty$, 
they are nontrivial critical points of $H_{\Omega}$. 
We therefore suggest that the dissipative system \eqref{dbMHD} 
can be applied to numerically compute MHD equilibria with nested flux surfaces. 

To prove proposition 4, let us first explain how system \eqref{dbMHD} is derived. 
Consider an $n$-dimensional Hamiltonian system 
\begin{equation}
\dot{z}^i=\mc{J}^{ij}\frac{\p H}{\p z^j}~~~~i=1,...,n,
\end{equation}
where $\bol{z}=\lr{z^1,...,z^n}$ are the phase space coordinates, $H=H\lr{\bol{z}}$ the Hamiltonian function (energy) of the system, and $\mc{J}^{ij}=-\mc{J}^{ji}$, $i,j=1,...,n$, the components of the Poisson tensor $\mc{J}$. 
The antisymmetry of $\mc{J}$ ensures conservation of energy according to $\dot{H}=H_i\dot{z}^i=\mc{J}^{ij}H_iH_j=0$. Here, the notation $H_i=\p H/\p z^i$ was used. 
Applying the Poisson tensor twice to the Hamiltonian $H$
has the opposite effect: the dynamical system
\begin{equation}
\dot{z}^i=\mc{J}^{ij}g_{jk}\mc{J}^{k\ell}\frac{\p H}{\p z^\ell},\label{dbdiss}
\end{equation}
where $g_{jk}=g_{kj}$ are the components of a symmetric positive semi-definite covariant 2-tensor, 
dissipates the energy $H$
according to
\begin{equation}
\dot{H}=-\mc{J}^{ji}\frac{\p H}{\p z^i}g_{jk}\mc{J}^{k\ell}\frac{\p H}{\p z^\ell}\leq 0,
\end{equation}
where the inequality follows from the positive semi-definiteness of the tensor $g$. 
The type of relaxation described by system \eqref{dbdiss} is 
called `double bracket dissipation' \cite{PJMdb1,PJMdb2}, because it is obtained by repeated application of the Poisson tensor $\mc{J}$ defining the Poisson bracket $\lrc{f,g}=f_i\mc{J}^{ij}g_j$. 
Double bracket dissipation 
can be used as an efficient method to compute nontrivial steady states of ideal dynamical systems such as the Euler equations because the resulting  equations preserve the Casimir invariants spanning the kernel of the Poisson tensor \cite{Furukawa,Vallis}. 

Now recall that system \eqref{MHD7} has the Hamiltonian form \eqref{EoMH}, which can be equivalently written as
\begin{equation}
{
\begin{bmatrix}
\frac{\p\Psi}{\p t}\\\frac{\p\Theta}{\p t}
\end{bmatrix}
=\mu_0A_0\lr{\Psi}\mc{J}_s
\begin{bmatrix}
\frac{\delta H_{\Omega}}{\delta\Psi}\\
\frac{\delta H_{\Omega}}{\delta\Theta}
\end{bmatrix}
=\mu_0 A_0\lr{\Psi}\begin{bmatrix}0&1\\-1&0
\end{bmatrix}
\begin{bmatrix}
\frac{\delta H_{\Omega}}{\delta \Psi}\\\frac{\delta H_{\Omega}}{\delta\Theta}
\end{bmatrix}},\label{dbdiss0}
\end{equation}
where 
\begin{equation}
\mc{J}_s=
\begin{bmatrix}
0&1\\-1&0
\end{bmatrix},
\end{equation}
is the contravariant symplectic Poisson 2-tensor. 
Double bracket dissipation 
for system \eqref{dbdiss0} 
can therefore be obtained 
according to
\begin{equation}
{
\begin{bmatrix}
\frac{\p\Psi}{\p t}\\\frac{\p\Theta}{\p t}
\end{bmatrix}
=\mu_0A_0\lr{\Psi}\mc{J}_s\Pi\mc{J}_s
\begin{bmatrix}
\frac{\delta H_{\Omega}}{\delta\Psi}\\
\frac{\delta H_{\Omega}}{\delta\Theta}
\end{bmatrix}
=\mu_0 A_0\lr{\Psi}
\begin{bmatrix}0&1\\-1&0
\end{bmatrix}
\begin{bmatrix}\sigma&0\\0&\gamma
\end{bmatrix}
\begin{bmatrix}0&1\\-1&0
\end{bmatrix}
\begin{bmatrix}
\frac{\delta H_{\Omega}}{\delta \Psi}\\\frac{\delta H_{\Omega}}{\delta\Theta}
\end{bmatrix}},\label{dbdiss1}
\end{equation}
 where the constant diagonal covariant 2-tensor 
 \begin{equation}
\Pi=
\begin{bmatrix}
\sigma&0\\0&\gamma
\end{bmatrix},
 \end{equation}
 which plays the role of $g$ in eq. \eqref{dbdiss}, 
serves the purpose of keeping the consistency of physical units. 
One can verify that evaluation of \eqref{dbdiss1} gives the anticipated system 
\eqref{dbMHD}. 

Let us now examine the main properties of system \eqref{dbMHD}. 
Setting $\p\Psi/\p t=0$ and $\p\Theta/\p t=0$ it is immediately clear that the resulting solution $\lr{\Psi,\Theta}$ is such that the magnetic field $\bol{B}=\nabla\Psi\cp\nabla\Theta$ solves equation \eqref{MHSdiss}. 
Now consider the rate of energy change,
\begin{equation}
\begin{split}
\frac{dH_{\Omega}}{dt}=&\int_{\Omega}\lrs{\frac{1}{\mu_0}\nabla\Psi\cp\nabla\Theta\cdot\lr{\nabla\frac{\p\Psi}{\p t}\cp\nabla\Theta+\nabla\Psi\cp\nabla\frac{\p\Theta}{\p t}}-\lambda\frac{\p\Psi}{\p t}}dV\\
=&
\frac{1}{\mu_0}\int_{\p\Omega}\lr{\frac{\p\Psi}{\p t}\nabla\Theta\cp\bol{B}+\frac{\p\Theta}{\p t}\bol{B}\cp\nabla\Psi}\cdot\bol{n}dS
\\&+\frac{1}{\mu_0}\int_{\Omega}\lr{\frac{\p\Psi}{\p t}\nabla\cp\bol{B}\cdot\nabla\Theta-\frac{\p\Theta}{\p t}\nabla\cp\bol{B}\cdot\nabla\Psi}dV-\lambda\int_{\Omega}\frac{\p\Psi}{\p t}\,dV.
\end{split}
\end{equation}
The boundary condition $\Psi={\rm constant}$ on $\p\Omega$ implies $\p\Psi/\p t=0$ there. Furthermore, recall that the unit outward normal to $\p\Omega$ satisfies $\bol{n}\cp\nabla\Psi=\bol{0}$. Using \eqref{dbMHD} 
we thus find that
\begin{equation}
\begin{split}
\frac{dH_{\Omega}}{dt}=&\frac{1}{\mu_0}\int_{\Omega}\lrs{\gamma A_0\lr{\mu_0\lambda-\nabla\cp\bol{B}\cdot\nabla\Theta}\lr{\nabla\cp\bol{B}\cdot\nabla\Theta}-\sigma A_0\lr{\nabla\cp\bol{B}\cdot\nabla\Psi}^2}dV\\&-\lambda\int_{\Omega}\gamma A_0\lr{\mu_0\lambda-\nabla\cp\bol{B}\cdot\nabla\Theta}dV\\=&
-\frac{1}{\mu_0}\int_{\Omega}A_0\lrs{\gamma\lr{\mu_0\lambda -\nabla\cp\bol{B}\cdot\nabla\Theta}^2+\sigma\lr{\nabla\cp\bol{B}\cdot\nabla\Psi}^2}dV\\=&-\frac{1}{\mu_0}\int_{\Omega}A_0^{-1}\lrs{\gamma\lr{\frac{\p\Psi}{\p t}}^2+\sigma\lr{\frac{\p\Theta}{\p t}}^2}dV\leq 0.\label{ineq}
\end{split}
\end{equation}
In the last passage we used the hypothesis $A_0\lr{\Psi},\gamma,\sigma>0$. 
This inequality 
shows that if the limit
\begin{equation}
\lim_{t\rightarrow +\infty}\frac{dH_{\Omega}}{dt}=0,
\end{equation}
exists, the corresponding solutions $\lr{\Psi_{\infty},\Theta_{\infty}}=\lim_{t\rightarrow +
\infty}\lr{\Psi,\Theta}$ 
of the dissipative system \eqref{dbMHD} 
are nontrivial solutions of \eqref{MHSdiss}. 
System \eqref{dbMHD} therefore provides a dynamical 
method to search for MHD equilibria that correspond to critical points 
$\delta H_{\Omega}\lrs{\Psi_{\infty},\Theta_{\infty}}=0$
of $H_{\Omega}$. 
Note that the key difference between the Hamiltonian system \eqref{MHD7} 
and the dissipative system \eqref{dbMHD} is that they allow 
the search for critical points in the 
state space $\mf{X}$ along different directions. 
More precisely, in the former case  
if a steady solution is found it corresponds to a critical point
having the same value of $H_{\Omega}$ as that associated with initial conditions due to conservation of energy,  
while in the latter case steady solutions 
will exhibit a lower value of $H_{\Omega}$ 
as a consequence of \eqref{ineq}.





\section{Construction of MHD equilibria by iteration}
In this section we propose an iterative scheme to construct nontrivial MHD equilibria 
of the type \eqref{MHSdiss} in $\Omega$. 
The iterative scheme is based on the observation that  
nontrivial steady solutions of system \eqref{MHD7} are given by the following system of 2 coupled PDEs for the Clebsch potentials $\lr{\Psi,\Theta}$ in $\Omega$: 
\begin{subequations}
\begin{align}
&\nabla\cdot\lrs{\nabla\Psi\cp\lr{\nabla\Theta\cp\nabla\Psi}}=0,\label{itMHS1}\\
&\nabla\cdot\lrs{\nabla\Theta\cp\lr{\nabla\Psi\cp\nabla\Theta}}=-\mu_0\lambda.\label{itMHS2}
\end{align}\label{itMHS}
\end{subequations}
In general, taking 2 arbitrary functions $\lr{\Psi_0,\Theta_0}\in\mf{X}$ they will not
solve system \eqref{itMHS}. The idea is to introduce new Clebsch potentials 
$\Theta_0,\Psi_1,\Theta_1,\Psi_2,\Theta_2,...,\Theta_{i-1},\Psi_i,...$ by iteratively solving equations \eqref{itMHS1} and \eqref{itMHS2} in the volume $\Omega$ starting from some function $\Psi_0\lr{\bol{x}}$, that is 
\begin{subequations}
\begin{align}
\nabla\cdot\lrs{\nabla\Psi_0\cp\lr{\nabla\Theta_0\cp\nabla\Psi_0}}&=0,\label{Th0}\\
\nabla\cdot\lrs{\nabla\Theta_0\cp\lr{\nabla\Psi_1\cp\nabla\Theta_0}}&=-\mu_0\lambda,\\
\nabla\cdot\lrs{\nabla\Psi_1\cp\lr{\nabla\Theta_1\cp\nabla\Psi_1}}&=0,\\
\nabla\cdot\lrs{\nabla\Theta_1\cp\lr{\nabla\Psi_2\cp\nabla\Theta_1}}&=-\mu_0\lambda,\\
\nabla\cdot\lrs{\nabla\Psi_2\cp\lr{\nabla\Theta_2\cp\nabla\Psi_2}}&=0,\\
\vdots&\\
\nabla\cdot\lrs{\nabla\Theta_{i-1}\cp\lr{\nabla\Psi_i\cp\nabla\Theta_{i-1}}}&=-\mu_0\lambda,\label{Psii}\\
\nabla\cdot\lrs{\nabla\Psi_i\cp\lr{\nabla\Theta_i\cp\nabla\Psi_i}}&=0,\label{Thetai}\\
\vdots& 
\end{align}\label{itex}
\end{subequations}
and so on, so that $\lim_{i\rightarrow\infty}\lr{\Theta_{i-1},\Psi_i}=\lr{\Theta_{\infty},\Psi_{\infty}}$ hopefully converges to 
a regular solution $\lr{\Theta_{\infty},\Psi_{\infty}}$ of system \eqref{itMHS}. 
Here, we observe that the first equation in the iteration \eqref{Th0} serves the purpose of determining $\Theta_0\lr{\bol{x}}$ from the `initial condition' $\Psi_0$.
In this context, a possible choice of boundary conditions is \eqref{bcpsith}. 
The following theorem states that, if at each step of the iteration
solutions of a prescribed regularity exist, then 
the iteration converges to a solution of \eqref{itMHS} (and \eqref{MHSdiss}) with the same regularity. 
More precisely we have 


\begin{theorem} 
Assume $\mu_0\lambda\neq 0$ and consider an iterative scheme in which equation \eqref{itMHS1} and equation \eqref{itMHS2} 
are solved alternately in $\Omega$ 
\begin{subequations}
\begin{align}
&\nabla\cdot\lrs{\nabla\Theta_{i-1}\cp\lr{\nabla\Psi_i\cp\nabla\Theta_{i-1}}}=-\mu_0\lambda,\label{it1}\\
&\nabla\cdot\lrs{\nabla\Psi_{i}\cp\lr{\nabla\Theta_i\cp\nabla\Psi_{i}}}=0,~~~~i=1,2,3,...,\label{it2}
\end{align}\label{iteration}
\end{subequations}
starting from a given pair $\lr{\Theta_0\lr{\bol{x}},\Psi_0\lr{\bol{x}}}\in\mf{X}$ such that
\begin{equation}
\nabla\cdot\lrs{\nabla\Psi_0\cp\lr{\nabla\Theta_0\cp\nabla\Psi_0}}=0,\label{Th02} 
\end{equation}
with $\nabla\Psi_0\cp\nabla\Theta_0\neq\bol{0}$. 
Suppose that during the iteration solutions exist and are nontrivial, i.e. $\nabla\Psi_i\cp\nabla\Theta_i\neq\bol{0}$ for $i\geq 1$.  
Further assume that the limit 
\begin{equation}
\lr{\Theta_{\infty},\Psi_{\infty}}=\lim_{i\rightarrow +\infty}\lr{\Theta_{i-1},\Psi_i},\label{lim}
\end{equation}
exists. Then, the pair $\lr{\Theta_{\infty},\Psi_{\infty}}$ solves equation \eqref{itMHS}. Furthermore, the vector field 
$\bol{B}=\nabla\Psi_{\infty}\cp\nabla\Theta_{\infty}$ defines a nontrivial MHD equilibrium 
solving equation \eqref{MHSdiss}.
\end{theorem}

\begin{proof}
We must show that the iteration procedure described above 
converges to 
the desired solution $\lr{\Theta_{\infty},\Psi_{\infty}}$ of \eqref{itMHS}.
We begin by noting that when a solution $\Psi_i$ of equation \eqref{it1} is found, 
the quantity
\begin{equation}
\tilde{\Psi}_i=\Psi_i-\int_{\Omega}\Psi_i\,dV,
\end{equation}
is also a solution of \eqref{it1}. Therefore, we can restrict the space of solutions $\mf{X}$ to those pairs $\lr{\Theta_{i-1},\Psi_i}$ such that
\begin{equation}
\int_{\Omega}\Psi_i\,dV=0.\label{avpsi}
\end{equation}
Next,  
we define
\begin{equation}
H_{ij}=H_{\Omega}\lrs{\Psi_i,\Theta_j}=\int_{\Omega}\lr{\frac{1}{2\mu_0}\abs{\nabla\Psi_i\cp\nabla\Theta_j}^2-\lambda\Psi_i}dV\geq 0,~~~~i,j=0,1,2,3,...,
\end{equation}
where $\lr{\Psi_i,\Theta_j}\in \mf{X}$, $i,j=0,1,2,3,...$, are determined in $\Omega$ iteratively according to equation \eqref{iteration}, and the last inequality follows from the property \eqref{avpsi}.  
Next, observe that setting 
$\Psi=\Psi_1$, $\Theta=\Theta_0$, 
$\delta\Psi=\Psi_0-\Psi_1$, and $\delta\Theta=0$ from equation \eqref{dH} one has 
\begin{equation}
H_{00}-H_{10}=\frac{1}{2\mu_0}\int_{\Omega}\abs{\nabla\delta\Psi\cp\nabla\Theta_0}^2dV\geq 0.
\end{equation}
Suppose that $H_{00}-H_{10}=0$. Then, $\Psi_1=\Psi_0+f\lr{\Theta_0}$ for some  function $f$ of $\Theta_0$. This implies that we have found a nontrivial solution 
of system \eqref{itMHS}
given by $\Theta=\Theta_0$ and $\Psi=\Psi_0$. This solution also defines an MHD equilibrium \eqref{MHSdiss} with $\bol{B}=\nabla\Psi_0\cp\nabla\Theta_0$.
Indeed, $\Theta_0$ is, by construction, a solution of \eqref{Th02}, while 
\begin{equation}
\nabla\cdot\lrs{\nabla\Theta_0\cp\lr{\nabla\Psi_1\cp\nabla\Theta_0}}=\nabla\cdot\lrs{\nabla\Theta_0\cp\lr{\nabla\Psi_0\cp\nabla\Theta_0}}=-\mu_0\lambda.
\end{equation}
We may therefore restrict our attention to the case $H_{00}>H_{10}$. 
In a similar manner, one finds 
\begin{equation}
H_{10}-H_{11}=\frac{1}{2}\int_{\Omega}\abs{\nabla\Psi_1\cp\nabla\delta\Theta}^2dV\geq0,
\end{equation}
 with $\delta\Theta=\Theta_0-\Theta_1$. Again, the case $H_{10}=H_{11}$ gives a nontrivial solution $\bol{B}=\nabla\Psi_1\cp\nabla\Theta_0$ of system \eqref{itMHS} since $\Theta_1=\Theta_0+g\lr{\Psi_1}$ for some smooth function $g\lr{\Psi_1}$. 
 Indeed, 
 \begin{subequations}
 \begin{align}
\nabla\cdot\lrs{\nabla\Theta_0\cp\lr{\nabla\Psi_1\cp\nabla\Theta_0}}&=-\mu_0\lambda\\
\nabla\cdot\lrs{\nabla\Psi_1\cp\lr{\nabla\Theta_1\cp\nabla\Psi_1}}&=\nabla\cdot\lrs{\nabla\Psi_1\cp\lr{\nabla\Theta_0\cp\nabla\Psi_1}}=0.
\end{align}
 \end{subequations}
Hence, either one finds a solution, or $H_{00}>H_{10}>H_{11}$. 
Repeating this procedure one may therefore construct a decreasing sequence 
\begin{equation}
H_{00}>H_{10}>H_{11}>H_{21}>H_{22}> ...>0,
\end{equation}
where the last inequality follows from the fact at each step of the iteration solutions are nontrivial (in particular $\nabla\Psi_i\cp\nabla\Theta_i\neq\bol{0}$ by hypothesis while $\nabla\Psi_i\cp\nabla\Theta_{i-1}\neq\bol{0}$ since $\mu_0\lambda\neq 0$ for all $i\geq 1$).
As explained above, the decreasing sequence may be interrupted if two contiguous steps possess the same energy, implying that a solution has been found after a finite number of iterations.  

Next, consider the pair $\lr{\Theta_{i-1},\Psi_i}$. 
We may quantify the degree at which  
$\lr{\Theta_{i-1},\Psi_i}$ fails to be a solution of system \eqref{itMHS} through the functional 
\begin{equation}
\Delta H_i=H_{ii-1}-H_{ii}=\frac{1}{2\mu_0}\int_{\Omega}\abs{\nabla\delta\Theta_i\cp\nabla\Psi_i}^2dV\geq 0,
\end{equation}
where $\delta \Theta_i=\Theta_{i-1}-\Theta_i$. Indeed, when $\Delta H_i=0$, one has $\delta\Theta_i=g\lr{\Psi_i}$ for some smooth function $g\lr{\Psi_i}$. 
Hence, 
\begin{subequations}
\begin{align}
\nabla\cdot\lrs{\nabla\Theta_{i-1}\cp\lr{\nabla\Psi_i\cp\nabla\Theta_{i-1}}}&=-\mu_0\lambda,\\
\nabla\cdot\lrs{\nabla\Psi_{i}\cp\lr{\nabla\Theta_{i-1}\cp\nabla\Psi_{i}}}&=\nabla\cdot\lrs{\nabla\Psi_{i}\cp\lr{\nabla\Theta_{i}\cp\nabla\Psi_{i}}}=0.
\end{align}
\end{subequations}
Now suppose that many iterations 
\eqref{iteration} are carried out, and that at each step $i$ the functional $\Delta H_i$ is evaluated.
We claim that
\begin{equation}
\lim_{i\rightarrow\infty}\Delta H_i=0.\label{limit}
\end{equation}
To see this, first note that $\Delta H_i<H_{00}<+\infty$ by construction. 
Next, let $\Delta H_n\in (0,H_{00})$ denote the value of $\Delta H_i$ at the $n$th iteration. Evidently, the number of times $\Delta H_i$ can be equal to or exceed  $\Delta H_n$ is limited to be at most
the natural number $N$ such that $N<H_{00}/\Delta H_n\leq N+1$. 
Hence, after a sufficiently large number $m$ of iterations we must have $\Delta H_i<\Delta H_n$ for all $i\geq m > n$. 
Now suppose that although the sequence $\Delta H_i$ is decreasing, its limit inferior does not reach zero, i.e. 
\begin{equation}
\liminf_{i\rightarrow\infty}\Delta H_i=\delta>0, 
\end{equation}
for some $\delta\in\mathbb{R}$. 
However, this is a contradiction, since 
one can always find a natural number $M$ 
such that $M\delta>H_{00}$. We must therefore conclude that
\begin{equation}
\liminf_{i\rightarrow\infty}\Delta H_i=0, 
\end{equation}
In particular, this implies that there exists some sufficiently large number of iterations $s$  
such that $0\leq \Delta H_s\leq\epsilon$
for any arbitrarily small $\epsilon >0$.
By the same argument as above, if $\Delta H_s>0$ the sequence $\Delta H_i$ can equal or exceed $\Delta H_s$ a finite number of times after which
\begin{equation}
\limsup_{i\rightarrow \infty}\Delta H_i=\Delta H_s\leq\epsilon.
\end{equation}
The limit of the sequence $\Delta H_i$ therefore belongs to $[0,\epsilon]$ for any arbitrarily small $\epsilon>0$. Equation \eqref{limit} thus follows.
This result implies 
\begin{equation}
\lim_{i\rightarrow\infty}\abs{\nabla\delta\Theta_i\cp\nabla\Psi_i}^2=0. 
\end{equation}
Recalling \eqref{lim}, it follows that    $\lim_{i\rightarrow\infty}\delta\Theta_i=g\lr{\Psi_{\infty}}$.
We therefore find 
\begin{subequations}
\begin{align}
\lim_{i\rightarrow\infty}\nabla\cdot\lrs{\nabla\Theta_{i-1}\cp\lr{\nabla\Psi_i\cp\nabla\Theta_{i-1}}}&=-\mu_0\lambda,\\
\lim_{i\rightarrow\infty}\nabla\cdot\lrs{\nabla\Psi_i\cp\lr{\nabla\Theta_{i-1}\cp\nabla\Psi_i}}&=\lim_{i\rightarrow\infty}\nabla\cdot\lrs{\nabla\Psi_{i}\cp\lr{\nabla\Theta_i\cp\nabla\Psi_{i}}}=0.
\end{align}
\end{subequations}
Hence, the pair
\begin{equation}
\lr{\Theta_{\infty},\Psi_{\infty}}=\lim_{i\rightarrow\infty}\lr{\Theta_{i-1},\Psi_i}, 
\end{equation}
defines a solution of system \eqref{itMHS}. 
Furthermore, $\mu_0\lambda\neq 0$ in \eqref{itMHS} implies that the vector field $\bol{B}=\nabla\Psi_{\infty}\cp\nabla\Theta_{\infty}$ is non-vanishing and that it is a    
solution of the MHD equilibrium equations \eqref{MHSdiss}. 

Finally, we observe that if a solution $\lr{\Theta_{n-1},\Psi_{n}}$ of system \eqref{itMHS} is found 
after a finite number of iterations $n$, successive iterations will simply return the same solution, e.g.  $\lr{\Theta_{n},\Psi_{n+1}}=\lr{\Theta_{n-1},\Psi_{n}}$. Hence, in this case $\lr{\Theta_{\infty},\Psi_{\infty}}=\lr{\Theta_{n-1},\Psi_n}$.
\end{proof}


\begin{remark}
If $\Omega$ is a hollow toroidal volume with boundary $\p\Omega$ 
corresponding to 2 distinct level sets of a smooth function $\Psi_0\in C^{\infty}\lr{\Omega}$, with $\nabla\Psi_0\neq\bol{0}$ in $\Omega$, and if
level sets of $\Psi_0$ foliate $\Omega$ with nested toroidal surfaces, 
theorem 1 of \cite{SY23} ensures that equation \eqref{Th02} always has a nontrivial solution $\Theta_0$ such that $\nabla\Psi_0\cp\nabla\Theta_0\neq\bol{0}$. 
Furthermore, the angle variable $\Theta_0$ is not unique, but solutions exist in the form $\Theta_0=M\mu+N\nu+\chi_0$, where $\mu,\nu$ are toroidal and poloidal angle variables, the integers $M,N$ determine the rotational transform of the vector field $\nabla\Psi_0\cp\nabla\Theta_0$, and the function $\chi_0\lr{\bol{x}}$ is single-valued. The same result applies when solving for $\Theta_i$ at any step \eqref{it2} of the iteration provided that $\Psi_i$ satisfies the same properties listed above for $\Psi_0$ in $\Omega$.   
\end{remark}

\begin{remark}
An argument analogous to that used in the proof of theorem 1 in \cite{SY23} 
shows that for a given angle variable $\Theta_{i-1}$ in \eqref{it1}, 
a solution ${\Psi_i}$ can be obtained by reducing equation \eqref{it1} 
to a 2-dimensional elliptic equation on each level set of $\Theta_{i-1}$ 
and by joining solutions corresponding to adjacent level sets.  
\end{remark}

\begin{remark}
In light of remarks 1 and 2 above, if one could show that at each step of the iteration the solutions $\Theta_{i-1}$ and $\Psi_i$, $i\geq 1$ preserve their properties (in particular, $\Theta_i$ remains an angle variable 
and $\Psi_i$ foliates the domain with nested toroidal surfaces) then, combining this result with theorem 1 proved in this section, one would have obtained a proof of the existence of MHD equilibria \eqref{MHSdiss} 
in hollow toroidal volumes of arbitrary shape.
In such construction, although no control is available on the form of the flux surfaces $\Psi_{\infty}$ within $\Omega$, one can conjecture that, 
if they exist, solutions $\bol{B}=\nabla\Psi_{\infty}\cp\nabla\Theta_{\infty}$ 
with different rotational transforms can be obtained by appropriate choice of the integers $M,N$ mentioned in remark 1.
\end{remark}





\section{Concluding remarks}

In this study, starting from the ideal MHD equations  \eqref{MHD1} and with the aid of Clebsch potentials, 
we derived a reduced set of equations \eqref{MHD2}, \eqref{MHD6}, as well as \eqref{MHD7} describing the nonlinear evolution of magnetic field turbulence 
in proximity of MHD equilibria \eqref{MHDE}. 
The ordering 
\eqref{ordering} used to arrive at these equations is appropriate for a plasma with small flow, small electric current, approximate flux surfaces, and slow time variation. This setting is expected to be relevant for stellarator plasmas. The same governing equations can be obtained under the more general ordering \eqref{ordering2} in which both the plasma flow and the electric current are not small. 
We showed that the reduced equations possess invariants. In particular, the closed system \eqref{MHD7} preserves magnetic energy, magnetic helicity, and total magnetic flux under suitable boundary and gauge conditions. 
Furthermore, it exhibits 
a noncanonical Hamiltonian structure with Poisson bracket \eqref{PB} and Hamiltonian \eqref{HOm} (proposition 1 of section 4).
Such Hamiltonian structure can be used to examine the stability properties of steady solutions: we found that MHD equilibria \eqref{MHSdiss} are nonlinearly stable against perturbations involving a single Clebsch potential in the sense of propositions 2 and 3 of section 5.
The Hamiltonian structure can also be applied to obtain a dissipative dynamical system \eqref{dbMHD} with the property that the Hamiltonian of the system is progressively dissipated as described by proposition 4 of section 6.  
System \eqref{dbMHD}, which comprises two coupled diffusion equations for the Clebsch potentials, thus provides a dynamical method to compute nontrivial MHD equilibria \eqref{MHSdiss} by minimizing the Hamiltonian \eqref{HOm}. 
We further proposed a second scheme to compute MHD equilibria \eqref{MHSdiss} based on the iterative solution of the two 
coupled equations \eqref{itMHS}. 
Here, theorem 1 shows that, if solutions exists at each step of the iteration, 
the process must converge toward a solution of system \eqref{itMHS} and thus to a nontrivial MHD equilibrium of the type \eqref{MHSdiss}.

The reduced equations derived in the present paper can be regarded as a toy model of turbulence that can be useful to assess dynamical accessibility and stability of MHD equilibria in physically relevant regimes. Furthermore,  they provide two practical approaches (a dissipative one and an iterative one) to numerically compute MHD equilibria. 
Finally, as outlined in the remarks at the end of section 7, we conjecture that the iterative scheme of section 7 may 
represent the basis for a mathematical proof of the existence of 
MHD equilibria with a non-vanishing pressure gradient in hollow tori of arbitrary shape (that is, configurations in which the boundary $\p\Omega$ is not invariant under some combination of Euclidean isometries). 


\section*{Acknowledgment}
N.S. would like to thank Z. Yoshida for useful discussion.  

\section*{Statements and declarations}

\subsection*{Data availability}
Data sharing not applicable to this article as no datasets were generated or analysed during the current study.

\subsection*{Funding}
The research of NS was partially supported by JSPS KAKENHI Grant No. 21K13851.
and 22H04936.

\subsection*{Competing interests} 
The authors have no competing interests to declare that are relevant to the content of this article.




\begin{thebibliography}{99}

\bibitem{Kruskal} M. D. Kruskal and R. M. Kulsrud, 
\ti{Equilibrium of a magnetically confined plasma in a toroid}, 
The Physics of Fluids \tb{1}, 4 (1958). 

\bibitem{LoSurdo} C. L. Surdo, 
\ti{Global magnetofluidostatic fields (an unsolved PDE problem)},
Int. J. Math. Math. Sci. \tb{9}, pp. 123-130 (1986).

\bibitem{YosYam} Z. Yoshida and H. Yamada, 
\ti{Structurally-unstable electrostatic potentials in plasmas}, 
Prog. Theor. Phys. \tb{84}, 203 (1990). 

\bibitem{GradRed} H. Grad,
\ti{Reducible problems in magneto-fluid dynamic steady flows},
Rev. Mod. Phys. \tb{32}, 4 (1960).

\bibitem{YosGiga} Z. Yoshida and Y. Giga, \ti{Remarks on spectra of operator rot},  Math. Z. \tb{204}, pp. 235-245 (1990).

\bibitem{GradToroidal} H. Grad, 
\ti{Toroidal containment of a plasma}, Phys. Fluids \tb{10}, 137 (1967).

\bibitem{Eden1} J. W. Edenstrasser, \ti{Unified treatment of symmetric MHD equilibria}, J. Plasma Phys. \tb{24}, pp. 299-313 (1980).

\bibitem{Eden2} J. W. Edenstrasser, \ti{The only three classes of symmetric MHD equilibria}, J. Plasma Phys. \tb{24}, pp. 515-518 (1980).

\bibitem{Hel} P. Helander, 
\ti{Theory of plasma confinement in non-axisymmetric magnetic fields}, 
Rep. Prog. Phys. \tb{77}, 087001 (2014). 

\bibitem{Rod21} E. Rodriguez, P. Helander, and A. Bhattacharjee, 
\ti{Necessary and sufficient conditions for quasisymmetry},
Phys. Plasmas \tb{27}, 062501 (2020).

\bibitem{Landre} M. Landreman and E. Paul, 
\ti{Magnetic fields with precise quasisymmetry for plasma confinement},
Phys. Rev. Lett. \tb{128}, 035001 (2022).

\bibitem{SatoSciRep} N. Sato, 
\ti{Existence of weakly quasisymmetric magnetic fields without rotational transform in asymmetric toroidal domains}, 
Scientific Reports \tb{12}, 11322 (2022).

\bibitem{MorrisonRMP} P. J. Morrison, 
\ti{Hamiltonian Description of the Ideal Fluid}, 
Reviews of Modern Physics \tb{70}, pp. 467-521 (1998).

\bibitem{Abdel} H. M. Abdelhamid, Y. Kawazura, and Z. Yoshida, 
\ti{Hamiltonian formalism of extended magnetohydrodynamics}, 
Journal of Physics A: Mathematical and Theoretical \tb{48} (23), 235502 (2015).

\bibitem{YosClebsch} Z. Yoshida, \ti{Clebsch parameterization: Basic properties and remarks on its applications}, J. Math. Phys. \tb{50}, 113101 (2009).

\bibitem{YosEpi2D} Z. Yoshida and P. J. Morrison, 
\ti{Epi-two-dimensional fluid flow: A new topological paradigm for dimensionality},  Phys. Rev. Lett. \tb{119}, 244501 (2017).

\bibitem{Wolt} L. Woltjer, 
\ti{A theorem on force-free magnetic fields}, 
Proc. Nat. Ac. Sci. \tb{44}, 6 (1958). 

\bibitem{Taylor} J. B. Taylor,
\ti{relaxation of a toroidal plasma and generation of reverse magnetic fields}, 
Phys. Rev. Lett. \tb{33}, 1139 (1974).

\bibitem{Taylor2} J. B. Taylor, 
\ti{Relaxation and magnetic reconnection in plasmas}, 
Rev. Mod. Phys. \tb{58}, 741 (1986).

\bibitem{YosMahPRL} Z. Yoshida and S. M. Mahajan,
\ti{Variational principles and self-organization in two-fluid plasmas}, 
Phys. Rev. Lett. \tb{88}, 9 (2002).

\bibitem{Morrison89} P. J. Morrison, 
\ti{Poisson Brackets for Fluids and Plasmas}, 
in Mathematical Methods in Hydrodynamics and Integrability in Dynamical Systems, eds. M. Tabor and Y. Treve, American Institute of Physics Conference Proceedings No. 88 (American Institute of Physics, New York, 1982) pp. 13-46.  

\bibitem{Little89} 
R. Littlejohn, 
\ti{Singular Poisson tensors}, 
M. Tabor and Y. Treve (Eds.), 
Mathematical Methods in Hydrodynamics and Integrability in Dynamical Systems, American Institute of Physics Conference Proceedings \tb{88}, 
American Institute of Physics, New York, pp. 47-66 (1982).

\bibitem{Olver} P. J. Olver, \ti{The Jacobi identity}, in  
Applications of Lie Groups to Differential Equations
(second ed.), Springer-Verlag, New York, pp. 436-445 (1993).

\bibitem{Holm} D. D. Holm, J. E. Marsden, T. Ratiu, and A. Weinstein, 
\ti{Nonlinear stability of fluid and plasmas equilibria}, 
Phys. Rep. \tb{123}, pp. 1-116 (1985).

\bibitem{Tronci} C. Tronci, E. Tassi, and P. J. Morrison, 
\ti{Energy-Casimir stability of hybrid Vlasov-MHD models},
J. Phys. A: Math. Theor. \tb{48}, 185501 (2015).

\bibitem{Rein} G. Rein, 
\ti{Non-linear stability for the Vlasov-Poisson system - the energy-Casimir method},
Mathematical Methods in the Applied Sciences \tb{17}, pp. 1129-1140 (1994).

\bibitem{Arnold} V. I. Arnold and B. A. Khesin, \ti{Stability criteria for steady flows},
in {Topological methods in hydrodynamics}, 
Springer, pp. 89-96 (1998).



\bibitem{PJMdb1} P. J. Morrison, 
\ti{A Paradigm for Joined Hamiltonian and Dissipative Systems}, Physica D \tb{18}, pp. 410-419 (1986).

\bibitem{PJMdb2} P. J. Morrison, \ti{Thoughts on Brackets and Dissipation: Old and New}, Journal of Physics: Conference Series \tb{169}, 012006 (12pp) (2009).

\bibitem{Furukawa} M. Furukawa, T. Watanabe, P. J. Morrison, and K. Ichiguchi, 
\ti{Calculation of large-aspect-ratio tokamak and toroidally-averaged stellarator equilibria of high-beta reduced magnetohydrodynamics via simulated annealing}, 
Phys. Plasmas \tb{25} 082506 (2018).

\bibitem{Vallis} G. K. Vallis, G. F. Carnevale, and W. R. Young, 
\ti{Extremal energy properties and construction of stable solutions of the Euler equations}, 
J. Fluid Mech. \tb{207}, pp. 133-152 (1989). 



\bibitem{Chodura} R. Chodura and A. Schl\"uter,
\ti{A 3D code for MHD equilibrium and stability}, 
J. Comp. Phys. \tb{41}, pp. 68-88 (1981).

\bibitem{Hir} 
S. P. Hirshman and J. C. Whitson, \ti{Steepestdescent moment method for threedimensional
magnetohydrodynamic equilibria}, 
 Phys. Fluids \tb{26}, 3553 (1983).


\bibitem{Darboux} M. de Le\'on, Methods of Differential Geometry in Analytical Mechanics (Elsevier, New York, 1989), pp. 250–253.

\bibitem{Pfeff} D. Pfefferl\'e, L. Noakes, and D. Perrella, 
\ti{Gauge freedom in magnetostatics and the effect on helicity in toroidal volumes}, 
J. Math. Phys. \tb{62}, 093505 (2021). 

\bibitem{Grad58} H. Grad and H. Rubin 
\ti{Hydromagnetic equilibria and force free fields},
Proc. 2nd United Nations Int. Conf. on the
Peaceful Uses of Atomic Energy \tb{31}, pp. 190-197 (1958).

\bibitem{SY23} N. Sato and M. Yamada,
\ti{Nested invariant tori foliating a vector field and its curl: toward MHD equilibria and steady Euler flows in toroidal domains without continuous Euclidean isometries}, 
J. Math. Phys. \tb{64}, 081505 (2023). 



 












\end{thebibliography}
\end{document}